\documentclass[runningheads]{llncs}
\usepackage{graphicx}
\usepackage{layout}
\usepackage{amsmath}
\usepackage{amssymb}
\usepackage{multirow}
\usepackage[colorinlistoftodos,prependcaption]{todonotes}
\usepackage{layout}
\begin{document}

\title{Revenue Maximization for Consumer Software: Subscription or Perpetual License?}
\titlerunning{Revenue Maximization for Consumer Software}
%
\author{Ludwig Dierks \and
Sven Seuken }
\authorrunning{L. Dierks and S. Seuken}
%
\institute{Department of Informatics, University of Zurich\\
\email{\{dierks,seuken\}@ifi.uzh.ch}}
\maketitle              
\begin{abstract}
We study the revenue maximization problem of a publisher selling consumer software. We assume that the publisher sells either traditional perpetual licenses,  subscription licenses, or both. For our analysis, we employ a game-theoretic model, which enables us to derive the users' equilibrium strategies and the publisher's optimal pricing strategy. Via extensive numerical evaluations, we then demonstrate the sizable impact different pricing strategies have on the publisher's revenue, and we provide comparative statics for the most important settings parameters.  Although  in practice, many publishers still only sell perceptual licenses, we find that offering a subscription license in addition to a perpetual license typically (but not always) leads to significantly higher revenue than only selling either type of license on its own.


\end{abstract}
%
%
\section{Introduction}
Consumer software, particularly video games, is a multi-billion dollar industry \cite{superdata,newzoo}.
Originally sold on physical media like CDs or DvDs, the rise of fast network connections has allowed software markets to become increasingly digital, eschewing any physical medium. This has brought with it a proliferation of new business models, for example microtransactions (i.e., the sale of many mini-upgrades for small amounts of money), lootboxes (i.e., randomized microtransactions \cite{chen2019loot}) or in-game advertisement \cite{burns2016sensitivity}.

In this paper, we analyze the revenue maximization problem of a software publisher who, while still focused on selling \textit{licenses} for his product, is open to do this either in the form of perpetual or subscription licenses.  Whereas  a classic perpetual license, once bought, allows a user access to the product for as long as he desires (or, in some more recent cases, as long as the publisher supports it), a subscription license only allows access to the product for as long as the user pays a (typically monthly) recurring fee.\footnote{Note that this is distinct from subscription services that give access to constantly changing bundles of products (e.g. Xbox Game Pass).} While in recent years, subscription licenses have have become common for cloud-based Software-as-a-Service offerings (where their main selling point is access to cloud hardware), most products that do not come with significant cloud hardware are still only sold through perpetual licenses (though some publishers have recently experimented with subscription models \cite{eu4,trackmania}). Since we are interested in the revenue effects of selling the product itself (as opposed to additional cloud-based features), we exclude any cloud-based synergy from our analysis.

For a publisher, offering a subscription model has a few obvious advantages compared to perpetual licenses: the barrier of entry gets reduced and natural product differentiation takes place between users depending on how long they are interested in the product. This potentially allows the publisher to obtain far higher revenue from some users than he could with perpetual licenses and may allow him to support the product for a longer period of time through a continuous revenue stream.

On the other hand, offering a subscription model also comes with certain disadvantages.  Users may stop subscribing once a product's novelty fades, while some users that would use the product for a long time with low intensity may not be willing to pay a recurring price at all. In addition, if the option to alternatively buy a perpetual license is also offered, ``market cannibalization'' between both offerings may occur. And lastly, but importantly, while publishers traditionally sell upgrades to keep their product up to date or expand its features, with a subscription model it is typically assumed that users always obtain access to the most recent version (not including optional micro-transactions).

In this paper, we take a game-theoretic approach towards analyzing the merit of offering subscription licenses instead of or in addition to perpetual licenses. Selling a product over some time horizon is fundamentally a question of revenue management \cite{gallego1994optimal,chen2019efficacy} and it is important to take user behavior into account, as users for example may delay a purchase to wait for a reduction in prices. In contrast to classic revenue management problems, software as a purely digital good has neither a limited stock nor marginal costs. Instead, the quality of the product in the eyes of users continuously decays. Furthermore, offering a subscription option and offering paid but optional upgrades both constitute forms of product differentiation \cite{mussa,moorthy,desai2001quality}, though again with very particular cost and utility structures that differ from the classic literature.  
 In the past, the revenue effects of subscriptions have been studied for some other domains like ancillary services of a repeatedly sold core product (e.g., additional baggage for airline tickets) \cite{wang2019pricing} or professional Software-as-a-Service offerings (where, importantly, subscriptions provide scalable hardware while buy options do not, and utilities take a very different form than for consumer software) \cite{rohitratana2012impact}.
 
To properly analyze the problem and capture all its particularities, we introduce a tailor-made model that takes the form of a two-step game. In the first step, the publisher chooses his pricing strategy; in the second step, the users act inside of a discrete time sub-game where they arrive and dynamically obtain and lose demand for the product. We prove that there are only five distinct classes of user equilibrium strategies, which significantly aids in our analysis. Based on this, we derive the publisher's revenue as a function of his pricing strategy and show that only offering a subscription option can never be optimal. We show that depending on the setting, either only offering perpetual licences or offering a subscription option in addition to perpetual licenses can be optimal, though for most software products, offering both options is likely to lead to the best possible revenue.  Through comparative statics we further evaluate the influence different setting  parameters have on the revenue of the various pricing  strategies.

\section{Model}
We model the problem as a two stage game. First, the publisher commits to a pricing strategy. Then, over $n_{max}$ timesteps users arrive to the system. Once a user has arrived to the system, he faces a game with multiple timesteps. We can model this as an infinite time horizon Markov Decision Process (MDP) where he can take actions and obtains rewards depending on the state he is in. 

\subsection{Publisher Model}
The publisher wants to sell a digital product. After $m$ timesteps, he offers an optional upgrade to the product.\footnote{We limit ourselves in the analysis to exactly \textit{one} upgrade after exactly $m$ time steps to simplify the exposition. Turning $m$ into a strategic variable, as well as extending our model to more than one upgrade or multiple price changes is straightforward.} The product has a base quality $q_1$, which the upgrade additively increases by $q_2$. While the base product is always available, the upgrade only becomes available from timestep $m$ onwards. Because of the digital nature of the good, we assume that the publisher has infinite supply and no marginal costs.
The publisher's strategy space consists of his choice of price vector $p=(p_1^{<m},p_1^{\geq m},p_2, p_S)$. Thus, he offers a menu of  options to his users: (1) buy the base product for a  one-time payment $p_1^n = p_1^{<m}$ (if bought in timestep $n<m$) or (2) for $p_1^n=p_1^{\geq m}$ (if bought in timestep $n\geq m$), (3) buy the upgrade for a one-time payment $p_2$ (only offered in timesteps $n\geq m$), or (4) subscribe for price $p_S$ per time step. A product that is bought can be used forever, but the upgrade needs to be bought separately. On the other hand, a subscription gives immediate access to all available upgrades, but the user loses access to the product when his subscription lapses.
 
 A publisher can choose  to not offer a buy or subscribe option by setting the corresponding price to infinity. The publisher's utility is equal to his expected revenue per user. 

\subsection{User model}
Users are identified by their state and type. 
A user's state is given by a tuple $\sigma = (d, o)$.  While each user starts out interested in obtaining access to the product, after using the product for some time this interest may vanish. The demand $d\in \left\lbrace0,1\right\rbrace$ denotes whether the user is still interested in the product, i.e., whether he obtains any utility for having access to it.  The ownership vector $o\in \left\lbrace0,1\right\rbrace^2$ denotes whether a user owns the base product, i.e.,  $o_1=1$ or the upgrade, i.e.,  $o_2=1$. 

A user's type is a tuple $\tau = (n_a, \delta,\gamma,v)$.
$n_a \in \left\lbrace 1,\ldots, n_{max}\right\rbrace$ denotes the timestep the user arrives into the system, i.e., the earliest time he could buy the product and is drawn from a distribution with pmf $f_a$.
$\delta \in (0,1)$ denotes the user's long term engagement factor with the product and is drawn from a distribution with pmf $f_\delta$. While any arriving user starts with demand $d=1$ for the product, in any timestep in which he uses the product he has a probability of $1-\delta$ to become uninterested and lose demand (setting $d=0$). A user who has lost demand no longer obtains utility from having access to the product. The release of the upgrade has the complementary probability $\delta$ to rekindle a lapsed user's interest and set $d=1$. 
$\gamma\in (0,1)$ denotes the quality decay factor of the user and is drawn from a distribution with pmf $f_\gamma$. While the product and the upgrade have qualities $q_i$ in the timesteps in which they become available, the realized quality of the product for the users decreases every timestep as hype and novelty fade and it slowly becomes outdated. 
The \textit{realized quality }of having access to $o$ at time $n$ is given by $q(o,\gamma, n)=  \sum_{i\in o} \gamma^{n-m_i} q$.
Lastly, $v\in [0,v_{max}]$ denotes the value a user has for a product of quality $1$, as long as he has demand. $v$ is drawn from a distribution with pdf $f_v$

A user's action space in any timestep consists of whether he subscribes $S\in \left\lbrace 0,1 \right\rbrace$ or buys the product or the upgrade $b\in \left\lbrace 0 ,1\right\rbrace^{2}$. Subscription gives a user immediate access to anything currently released, i.e., $o_S^n = [1,0]$ if $n<m$ and $o_S^n = [1,1]$ otherwise, while buying $b$ gives ownership of the bought product, i.e., changes the ownership vector from $o$ to $o'= max(o,b)$.

The \emph{normalized immediate reward} $w_n$ of a user of type $\tau$ in timestep $n$ and state $(d, o)$ is given by
$w_n(S,b,\tau, \sigma) =d  \left((1-S)q(max(o,b), \gamma, n) +S q(o_S^n, \gamma, n)\right)$, while his \emph{immediate payment} is given by $\rho_n(S,b,\tau, \sigma,p) = p_S S + p_1^n b_1 + p_2 b_2$.
His overall \emph{immediate utility} in timestep $n$ is therefore given by 
$u_n(S,b,\tau, \sigma, p) =v w_n(S,b,\tau, \sigma)  - \rho_n(S,b,\tau, \sigma,p)$.

A strategy $\alpha(n, \sigma):\mathbb{N} \times \left\lbrace 0,1 \right\rbrace \times  \left\lbrace 0,1 \right\rbrace^2 \rightarrow \left\lbrace 0,1 \right\rbrace \times  \left\lbrace 0,1 \right\rbrace^2$ maps timesteps and user states to actions.
The \emph{normalized expected reward} for playing strategy $\alpha$ is given by 
\begin{align}w(\alpha,\tau) = \sum_{n=n_a}^\infty \sum_{\sigma'} P(\sigma_n= \sigma'|\alpha, \tau, \sigma) w_n(S,b,\tau, \sigma),\end{align} 
where $P(\sigma_n= \sigma'|\alpha, \tau, \sigma)$ denotes the probability of the user being in state $\sigma'$ during timestep $n$ given $\alpha, \tau, \sigma$.
Similarly, the expected payment is given by  
\begin{align}\rho(\alpha,\tau,p) = \sum_{n=n_a}^\infty \sum_{\sigma'} P(\sigma_n= \sigma'|\alpha, \tau, \sigma) \rho_n(S,b,\tau, \sigma,p).\end{align} 
 A user's overall \emph{expected utility} with strategy $\alpha$ is consequently given by 
\begin{align}u(\alpha,\tau) = v w(\alpha,\tau)- \rho(\alpha,\tau,p).\end{align} 

\section{User Equilibrium Strategies}\label{Sec:MDP}
Before we can analyze the publisher's revenue, we first need to determine how users would react to any given publisher strategy. Since the supply of software is unlimited and since we do not model any social effects, the utility of a given user is independent of the strategies of the other users. We can therefore find the equilibrium strategy of each user by solving his individual MDP in isolation.  

For any given user type $\tau$, we could directly do this through backward induction. But doing so on a user's full strategy space is computationally very costly, even for one user. As we later need to compute the optimal strategies for each user type to calculate the publisher's revenue, we now show that the optimal strategies for each user type can only come from a small set of possible strategies. Note that throughout this section, most expressions  depend on the publisher strategy $p$. For the sake of readability, we keep this dependency implicit and omit $p$ wherever doing so does not cause confusion.

A first important observation for identifying potentially optimal user strategies is that conceptually, the MDP for each user type can be seen to consist of two distinct parts: anything that happens before the upgrade is released in timestep $m$ and anything that happens afterwards. For notational clarity and ease of exposition, we split user strategies in this manner, i.e., setting $\alpha=(\alpha_1,\alpha_2)$, where $\alpha_1$ denotes a user's strategy before $m$ and $\alpha_2$ denotes his strategy beginning from timestep $m$.  

Before timestep $m$, a user's actions are restricted to buying the base product, doing nothing or subscribing. Note that subscribing in any timestep yields the same immediate reward as owning the base product. Given this, we easily obtain the following result that shows that there is only one potentially optimal strategy  that involves  buying the product and one potentially optimal strategy that involves subscribing (though possibly for zero timesteps). Overloading notation, we denote these by $\alpha_1=b$ and $\alpha_1 = s$, respectively.  

\begin{lemma}\label{LEM:1}\begin{enumerate}
                \item 
        For a user of type $\tau=(n_a, \delta, \gamma, v)$ that buys the base product before timestep $m$, the optimal strategy $\alpha_1=b$ has him buy in the timestep $n_a$ he arrives and not subscribe in any timestep $n<m$. 
        
\item   For a user of type $\tau= (n_a, \delta, \gamma, v)$ that does not buy the base product before timestep $m$ and plays some $\alpha_2$ from timestep $m$ onwards, there exists $n_1^{\alpha_2,\tau}$ such that the optimal strategy $\alpha_1=s$ has him  subscribe in any timestep $n <n_1^{\alpha_2,\tau}$ where he has demand and no other timesteps.  It holds that $n_1^{\alpha_2,\tau}$ is the smallest $n\geq0$ for which it holds
        \begin{align}
        v <& \begin{cases}      
\frac{p_S-(1-\delta)^2 \rho((\alpha_s,\alpha_2),\tau',p) }{q([1,0],\gamma, n)-(1-\delta)^2 w((\alpha_s,\alpha_2),\tau')} &\text{if } n < m
        \\
        \infty& \text{if } n = m
        \end{cases}
        \end{align}
        where $\tau'= (m, \delta, \gamma, v)$ (i.e., $\rho((\alpha_s,\alpha_2),\tau', p)$ and $w((\alpha_s,\alpha_2),\tau')$ are the reward and payment the user would obtain if he arrived in timestep $m$). 
        \end{enumerate}
\end{lemma}
\begin{proof}\begin{enumerate}
                \item 
        The first statement follows directly by noting that when buying in a later timestep before $m$, a user's additional realized value compared to subscribing or doing nothing decreased, but his payment did not decrease. He will therefore buy in the first possible timestep and afterwards can not obtain any additional reward through subscribing.   
        
        \item A subscribing user obtains immediate utility $v q([1,0],\gamma, n) - p_S$ in any timestep where he is subscribed, which decreases in $n$. Thus, if he does not subscribe in any timestep $n_1^\tau$, then he will also not subscribe in any later timestep before $m$. Additionally, in any timestep where he is subscribed, he has probability $(1-\delta)$ to lose demand. If he loses demand before timestep $m$, then he still has probability $\delta$ to regain demand in timestep $m$. This means that subscribing in timestep $n<m$ decreases the expected utility he obtains after timestep $m$ by a factor of $\delta^2$.  As he does not buy before $m$, his utility from timestep $m$ onward, if $d=1$, is the same as if he had arrived in timestep $m$, i.e., as if his type was $\tau' = (m, \delta, \gamma, v)$. The overall change in expected utility for subscribing in the largest timestep $n$ he subscribes is thus given by   \begin{align}
        v q([1,0],\gamma, n) -p_s - v (1-t)^2 w((\alpha_s,\alpha_2),\tau')+  (1-t)^2 \rho((\alpha_s,\alpha_2),\tau',p) 
        \end{align}
    and $n_1^\tau$ is the first timestep for which this change would be negative. 
        \end{enumerate}
\end{proof}

Given Lemma \ref{LEM:1} and a strategy $\alpha_2$ to play from timestep $m$ onwards, there are only two potentially optimal strategies before timestep $m$: buy in the timestep a user arrives, denoted by  $\alpha_1=b$, or subscribe until timestep $n_1^{\alpha_2,\tau}$, denoted by $\alpha_1 = s$.
To determine which of these two strategies is optimal, we must next take the user's strategy after the upgrade release into account.

After the upgrade releases, the user's space of potentially optimal strategies grows slightly.
In addition to buying the upgrade, and if not yet owned, the base product ($\alpha_2=b$) or not buying anything ($\alpha_2 = s$), a user might also decide to only buy the base product, but not the upgrade ($\alpha_2 = b_b$). This can for example happen if the base product is heavily discounted after timestep $m$, but the price of the upgrade is set very high. For such a user it might be optimal to first subscribe for a few timesteps, before buying the base product.

\begin{lemma}\label{LEM:2}\begin{enumerate}
                \item 
                For a user of type $\tau$ who buys the upgrade (and, if not yet owned, the base product), the optimal strategy $\alpha_2=b$ is to buy as soon as possible (i.e., in timestep $max(n_a, m)$) and to not subscribe in any timestep $n\geq m$. 
                
                \item  For a user of type $\tau$ with ownership vector $o$ who does not buy anything after timestep $m$, for the optimal strategy $\alpha_2=s$ there exists a timestep $n_2^{o, \tau}\geq m$ such that he subscribes in any timestep $n$ with $m \leq n< n_2^{o,\tau}$ where he has demand $d=1$ and subscribes in no timestep $n\geq n_2^{o, \tau}$. It holds that $n_2^{o,\tau}$ is the smallest $n \geq 0$ for which 
                \begin{align}
                v <& \frac{p_S}{q([1,1]-o,\gamma, n)}.
                \end{align}

                \item For a user of type $\tau$ that only buys the base product after timestep $m$ (and never buys the upgrade),  for the optimal strategy $\alpha_2=b_b$ there exists a timestep $n_3^{\tau}$ such that he subscribes in and only in any timestep $n$ with $m\leq n<n_3^{\tau}$ where he has demand and buys in timestep $n_3^{\tau}$ (if he still has demand).
                It holds that $n_3^{\tau}$ is the smallest $n \geq 0$ for which 
                \begin{eqnarray}
             v < \frac{p_S-(1-\delta) p_1^{\geq m}}{q([0,1],\gamma, n_3^{\tau})}.
                \end{eqnarray}
        \end{enumerate}
\end{lemma}

\begin{proof}
        
        \begin{enumerate}
                \item  Follows analogously to the proof of the statement for buying before $m$ in Lemma \ref{LEM:1}.
                \item  Follows analogously to the proof of the statement for users who do not buy in Lemma \ref{LEM:1}, though the users value for subscribing depends on whether they already own the base product or not. 
                
                \item   A user who does not buy the upgrade, but who does buy the base product after timestep $m$, optimally does so in the first timestep where he does not subscribe. Not doing so only decreases his reward but not his payment.            
                When subscribing in timestep $n$, such a user obtains an additional value of $q([0,1],\gamma, n)v$ over just owning the base product and makes a payment of $p_S$.  
                He further has a probability of $(1-\delta)$ to lose demand before the next timestep, in which case he does not buy at all, saving  (in expectation) $(1-\delta)p_1^{\geq m}$. This means that a user's utility for subscribing in the highest timestep in which he subscribes is given by $q([0,1],\gamma, n)v-p_S+(1-\delta)p_1^{\geq m}$. As this utility decreases in $n$, the user either does not subscribe at all (i.e. buys in timestep $m$ or when he arrives) or there exists a smallest timestep $n_3^{\gamma, v}$ for which his added utility for subscribing becomes negative and in which he buys.  Note that here the user wants to buy by assumption, even though doing so might not be optimal anymore in timestep $n_3^{\gamma, v}$. If that is the case, then not buying the base product and subscribing until $n_2^{[0,0], \tau}$ (i.e., playing $\alpha_2 = s$) would be better.
        \end{enumerate}

\end{proof}

Given Lemma \ref{LEM:2}, there are only three potentially optimal strategies beginning in timestep $m$: (1) buy the upgrade (and, if not owned yet, the base product) once the upgrade releases in $m$ (or once the user arrives if $n_a>m$), (2) not buy anything and subscribe until timestep $n_2^{o,\tau}$, or (3) subscribe before timestep $n_3^{\tau}$, then buy the base product in time step $n_3^{\tau}$. We denote these by $\alpha_2=b$, $\alpha_2 = s$ and $\alpha_2 = b_b$ respectively.

Taken together, Lemmas \ref{LEM:1} and \ref{LEM:2} describe all potentially optimal strategies for any user.

\begin{proposition}\label{prop:optS}
        It maximizes the expected utility of a user of type $\tau$ to play strategy
        \begin{align}
        \alpha^\ast_{\tau, p} = argmax_{\alpha_1 \in \left\lbrace b,s\right\rbrace, \alpha_2 \in \left\lbrace b,s, b_b\right\rbrace} v w((\alpha_1, \alpha_2), \tau) -\rho((\alpha_1, \alpha_2), \tau, p) 
        \end{align}
\end{proposition}
\begin{proof}
        Follows directly from combining Lemmas \ref{LEM:1} and \ref{LEM:2}.
\end{proof}

Results on how to calculate the reward and payments for each strategy in an efficient way can be found in  Appendix \ref{AP1}.

%
%
%
%
%

\section{Publisher Revenue}
 Given the results for the optimal user strategies from Section \ref{Sec:MDP}, we can now give a relatively simple expression for the publisher's revenue.

\begin{proposition}
Given strategy $p= (p_1^{<m}, p_1^{\geq m}, p_2, p_S)$, the publisher's expected revenue per user is given by 
        \begin{align}
        \pi(p) = &\sum_{n_a}\sum_\delta \sum_{\gamma } \int_v   \rho\left(\alpha^\ast_{(n_a,\delta,\gamma,v),p},(n_a,\delta,\gamma,v),p\right) f_a(n_a)f_\delta(\delta)f_\gamma(\gamma)f_v(v) dv
        \end{align}
where   $\alpha^\ast_{(n_a,\delta,\gamma,v),p}$ is given by Proposition \ref{prop:optS}
\end{proposition}
\begin{proof}
 Follows directly by taking the optimal user strategies $\alpha^\ast_{(n_a,\delta,\gamma,v),p}$ for each type $\tau$ as given by  Proposition \ref{prop:optS} and taking the expectation over all types.
\end{proof}

While this expression for the publisher's revenue can easily be evaluated for a given $p$ (employing the results in Appendix \ref{AP1}) it is highly non-convex in $p$ and has many local maxima: Changing prices affects different users differently and usually increase the payments of some users but reduces the payment of other users.  Consequently, first-order derivative tests are very bad indicators for whether a given price vector $p$ is close to optimal. 

Similarly, whether only offering an (optimal) buy or subscription option would lead to better revenue for the publisher depends on the complex interaction of all parameters of the setting. We can find settings where either strategy yields higher revenue for the publisher and there does not seem to exist a simple condition for deciding which is optimal without solving for the optimal price $p$. 
But what we can say is that, while only offering a buy option is optimal in at least some settings, \emph{only} offering a subscription option can never be as good as offering both options to users.

\begin{proposition}While only offering a buy option can be revenue optimal in some settings, only offering a subscription option is never revenue optimal.
\end{proposition}
\begin{proof}
        To see that only offering a buy option can be optimal, consider a setting without a later upgrade (i.e., $m=1$) and only a single user type $\tau$. Since the expected reward for owning is higher than the expected reward for subscribing up to any finite timestep, users are willing to pay more for perpetual licenses than for subscribing. Since there is only one user type and by Lemma \ref{LEM:2} buying users buy in the timestep they arrive, adding an additional subscription option can not extract additional revenue.
 
        To see that only offering a subscription option can never be optimal, assume some $p$ with $p_S<\infty$ and $p_1^{<m}=p_1^{\geq m}=p_2=\infty$ is optimal. Given $p$, let $\tau_{max}$ be the set consisting of the user types that make the highest expected payment after timestep $m$. Denote this payment by $\rho_{max}^{>m}$.  Since the reward for buying in timestep $m$ is always strictly higher than the reward for subscribing from $m$ to any finite timestep, users of type $\tau_{max}$ would be willing to pay $\rho_{max}^{>m}+\epsilon$ for owning the product from timestep $m$ onwards. Setting $p_1^{\geq m}= \rho_{max}^{>m}+\epsilon$, $p_2 = 0$ for $\epsilon >0$ small enough therefore leads to users in a neighborhood around $\tau_{max}$ buying in timestep $m$ and paying strictly more than $\rho_{max}^{>m}$, while no user pays less. Consequently,  $p = (\rho_{max}^{>m}+\epsilon, \infty,0, p_S)$ yields higher revenue for the publisher and only offering subscriptions cannot be revenue optimal. 
\end{proof}

\section{Numerical Evaluation}
To better understand when offering subscriptions can increase a publisher's revenue and by how much it typically does so, we now present a number of numerical examples and comparative statics. For each example, we give the optimal revenue for the \textit{optimal} prices for a publisher  who either (1) only offers a buy option (Opt(Buy)), (2) only offers a subscription option (Opt(Sub)), (3) offers both options (Opt(Both)), or (4) offers both options, but restricts the buy prices of perpetual licenses to those that would be optimal without subscription option (Opt(Both $\vert$ Opt(Buy))).

Since the publisher's revenue as a function of his price vector has many local maxima, all optimal publisher strategies in the following are calculated using a best-of-15 differential evolution search \cite{storn1997differential}. As this is a stochastic search, full optimality for all data points cannot be guaranteed. Those imprecisions can be seen as small local fluctuations of the function values in our plots.

\subsection{Experimental Setup: Video Games}
For our numerical analysis to have merit, we need to choose realistic settings parameters.
We chose the domain of video games for our numerics, because some user data as well as pricing data is available for this domain, which we can use to 
inform our choice of parameters and distributions.\footnote{To the best of our knowledge, no comprehensive study of how users in this domain are typically distributed is available, which is why we have to base our parameter choices on a rough analysis of some limited datasets that we have access to.} 
 In this section, we describe some basic insights we have obtained from the available data and how they have influenced our choice of parameters.  We bought a dataset from the website \textit{Steam Spy} \footnote{https://steamspy.com/about}, which collects publicly available data from the large video game storefront \textit{Steam}\footnote{https://store.steampowered.com/about/} and uses it to statistically estimate the number of owners of video games over time and what percentage of them actively used the game recently (i.e., in the last two weeks). 
While this data is prone to estimation errors and only contains limited information about those users that bought the games, it still allows us to make a number of general observations to help find reasonable distributions. In the following, we describe the general insights we have obtained from analyzing this data.\footnote{The purpose of this exercise was to find \emph{reasonable} parameter settings and distributions for the comparative statics. A detailed empirical analysis (e.g., fitting a statistical model to the data) is beyond the scope of this paper.} We provide plots for some representative games  in Appendix \ref{AP2}. 
\begin{enumerate}
        \item For most games, while the game is supported, its monthly sales  stay roughly constant as long as the price of the game keeps decreasing and afterwards drops off relatively slowly. This suggests that the arrival rate of new users is roughly constant, though the quality decays, which in turn is counteracted by price drops. The exception to this is the release month (+/- 1-2 weeks), which for big releases can have $3-10$ times as many sales, as many users effectively arrived before the game's release but could not buy it yet. Additionally, the release of bigger upgrades often boosts sales of the base game, which can be explained by users in the system holding off their purchase until the upgrade releases. Denoting the overall arrivals during the first timestep by $x_a$, we set the arrival distribution $f_a$ to be 
        \begin{align}
        f_a(n_a)= \begin{cases}
        \frac{x_a}{x_a+n_{max}-1}& \text{if } n_a = 1\\
        \frac{1}{x_a+n_{max}-1}& \text{if } 1<n_a < n_{max}
                \end{cases}
                \end{align}

        While we will later vary $x_a$, we choose $x_a=5$ as the standard value for most of the section. 
        \item While games on Steam usually do not change their base price, most games get regularly (i.e., typically every few weeks) discounted for a limited period of time, and most users buy during those discount periods. Taking this into account, the price of most games effectively drops by $40\%-60\%$ during the first year. As the number of new owners per timestep stays roughly constant and assuming the arrival rate of users is constant, this suggests that the quality of most games for users that do buy on average decays at a rate around $10\%$ per month, i.e., $\gamma = 0.9$. As not much more distributional information is available, for simplicity we assume for our numerics that $\gamma \in \left\lbrace 0.85, 0.9, 0.95\right\rbrace$.  We let $x_\gamma$ denote the probability that $\gamma$ is $0.9$ and set the $\gamma$ distribution $f_\gamma$ to be 
        $f_\gamma(0.9) = x_\gamma$ and $f_\gamma(0.85) = f_\gamma(0.95) = \frac{1}{2}(1-x_\gamma)$.
        While we will later vary $x_\gamma$, we choose $x_\gamma= 0.8$ as the standard value for most of the section. 
        \item While Steam Spy does not contain data on whether a user played a game after a certain date, it does contain data for the percentage of users who own the game and have played it in the last two week. Using this as a proxy for the percentage of users with demand $d=1$, we see that the percentage of users who stop playing after a month for most games varies between $40\%$ and $80\%$ percent. For most games, once the percentage of active users has reached about $20\%$, it starts to only fall very slowly. Accounting for the fact that there are constantly new users buying the game, it is reasonable to assume that $20\%$ of users have a $90\%$ probability to not lose demand in each timestep, i.e., $\delta = 0.9$. We therefore settle on a simple two-type distribution of long-term and short-term users. Denoting the probability that a short-term user does not lose demand after one timestep by $x_\delta$, we obtain $f_\delta(x_\delta) = 0.8$ and $f_\delta(0.9) = 0.2$.
        While we will later vary $x_\delta$, we choose $x_\delta= 0.5$ as the standard value used for most of the section.
                \item The dataset does not contain much information that would allow us to estimate the distribution of user values. We therefore simply set the user values to be distributed according to a normal distribution with mean $\mu = 25$ truncated to $[0,50]$. While we will later vary the standard deviation $\sigma$ of $f_v$, we choose $\sigma = 10$ as the standard value for most of the section.
\end{enumerate}
Additionally, for the numerical analysis, we set $q_1= 1$, $q_2=0.5$, $m=6$ and $n_{max}=12$.

\subsection{Base case}

        \begin{table*}[t!]
                \small
                
                \centering
                
                \begin{tabular}{|c|c|c|c|c||c|c|c|}
                        \hline
                 & $p_1^{<m}$ &  $p_1^{\geq m}$& $p_2$& $p_S$ & Revenue & User Welfare &Overall Welfare  \\
                        \hline \noalign{\smallskip                              }
                        
                        \hline 
                        Opt(Buy) & $45.82$  &$21.8$ &$18.06$&$\infty$ & $31.42$ &$33.23$&$64.65$\\
                \hline  
                                Opt(Sub) & $\infty$  &$\infty$ &$\infty$&$14.66$ & $33.54$ &$24.07$&$57.61$\\
                                \hline 
                                        Opt(Both) & $96.98$ &$35.19$& $47.96$&$17.71$ &$37.88$ &$24.39$& $62.27$\\
                        \hline 
                                                Opt(Both $\vert$ Opt(Buy)) & $45.82$  &$21.8$ &$18.06$&$18.4$ & $31.82$ &$33.89$ &$65.71$\\
                        \hline \noalign{\smallskip}                     
                \end{tabular}
                \caption{Results for four different publisher strategies for the base case}  \label{table1}
        \end{table*}
        
In this section, we discuss the numerical results for the base case with $x_a = 5, x_\gamma = 0.8, x_\delta = 0.5$ and $\sigma = 10$. This parametrization roughly corresponds to a typical game based on our analysis of the Steam Spy data. The results for each type of publisher strategy are summarized in Table \ref{table1}.

The best attainable revenue for a publisher who only wants to sell perpetual licenses without offering a subscription option (i.e., Opt(Buy)) is $\pi(p)=31.42$. As this revenue is attained with a relatively low price  (that is further discounted roughly $50\%$ once the upgrade releases), the publisher ensures that most users buy his product. This leaves the users with relatively high utility for owning the game  and thus the users' social welfare (i.e., the expected utility of a randomly drawn user) is relatively high at $33.23$.

If the publisher would instead only offer a subscription option (i.e., Opt(Sub)), his expected revenue per user increases to $\pi= 33.54$, a substantial $6.7\%$ increase over only offering the buy option. This is possible because, when only offering a buy option the publisher had to set a relatively low price to attract users that are only interested in playing the game for a short time as well as users that want to play it for a long time. The subscription option on the other hand automatically price discriminates between those user types and extracts more revenue from long-term users.  Consequently, this revenue increase comes at the cost of the user welfare, which decreases substantially. Unfortunately, this is not simply a transfer of utility from the users to the publisher, as the users welfare decreases more than the publisher's revenue increases. This loss is caused because users whose perceived quality of the game decayed too much stop subscribing, even though they would like to continue playing.  They are simply not willing to pay the subscription price anymore. This  decreases the system's overall welfare (i.e., the sum of revenue and user welfare) by $11\%$.

If the publisher decides to offer both perpetual and subscription licenses (i.e., Opt(Both)), then he can further increase his revenue to $\pi= 37.87$, an additional $12.9\%$ increase over only offering a subscription option and a staggering $20.5\%$ increase over only selling perpetual licenses.
This revenue increase requires that all prices rise compared to when only one of the two alternative license types are offered. While the subscription price only moderately increases to $p_S= 17.71$, the buy prices roughly double to $p_1^{<m}= 98.98, p_1^{\geq m}= 35.19, p_2=47.96$. 
Interestingly, with these prices, any user that subscribes for $3$ or less timesteps before the price change and then buys the discounted base product still pays less overall than if he would have bought the base product directly. Consequently, we see that $52.9\%$ of users that arrive in timestep $1$ only subscribe at first, but ultimately buy a perpetual licenses once it is discounted (as long as they still have demand). Only $15.3\%$ of users that arrive in timestep $1$ are willing to directly buy the game at its high starting price.  
This truly splits the user base into two parts. First there are casual users that subscribe for a few timesteps and often pay even less than they would with the buy option. Then there are power users that plan to use the game for a long time, often longer than they would be willing to subscribe, and are thus willing to pay the increased buy price. Consequently, despite the notable revenue increase, user welfare does not decrease further and, compared to only offering a subscription option, even slightly increases to $24.4$. While this is still notably lower than the user welfare when only offering perpetual licenses, the system's overall welfare (i.e., the sum of revenue and user welfare) now is only about $3.7\%$ lower, showing that most of the user welfare gets transferred to the publisher instead of being lost.

Lastly, we analyze whether the publisher can increase his profit by offering a subscription option without changing the buy prices of perpetual licenses (i.e., Opt(Both$\vert$ Opt(Buy))).
This would guarantee that no user is worse off than when only perpetual licenses are offered, which is an important consideration when there are  competing products. Any revenue increase under such a pricing model has to come from additional users that do not buy perpetual licenses even when no subscription option is offered. 
If the publisher fixes his buy prices at  $p_1^{<m}= 45.82, p_1^{\geq m}= 21.8, p_2=18.06$, the optimal subscription price is $p_S =18.4$. Considering that the publisher wants to attract additional users, it might seem counterintuitive that $p_S$ here is even higher than the subscription price that was optimal combined with the far higher optimal buy price. This effect occurs because the lower the buy price, the more users with low long-term engagement are willing to buy when no subscription option is offered. But those same users readily switch to subscribing and pay even less when $p_S$ is low, increasing market cannibalization.  
Consequently, the attainable revenue increase is comparatively modest, with the publisher obtaining at most $\pi = 31.82$, an increase of $1.3\%$. While this pales compared to the potential revenue increase with fully optimized prices, it can still constitute an additional revenue of hundredth of thousands or even millions of dollars for large releases. Importantly, since this pricing strategy, by construction, leads to a Pareto improvement for the users, user welfare also slightly increases.

\subsection{Comparative Statics}
We now study how varying the setting parameters $x_\delta, x_\gamma, \sigma, x_a$ affects the publisher's revenue under his four different strategies. In Figures \ref{FIG:vart} to \ref{FIG:vara}, we present comparative statics for how the optimal revenue of each type of publisher strategy changes in relation to the revenue of only offering a buy option (which we normalize to 1). 

       \begin{figure}[h!]
        \centering%
        \begin{minipage}{0.50\textwidth}
                \includegraphics[width=0.99 \textwidth]{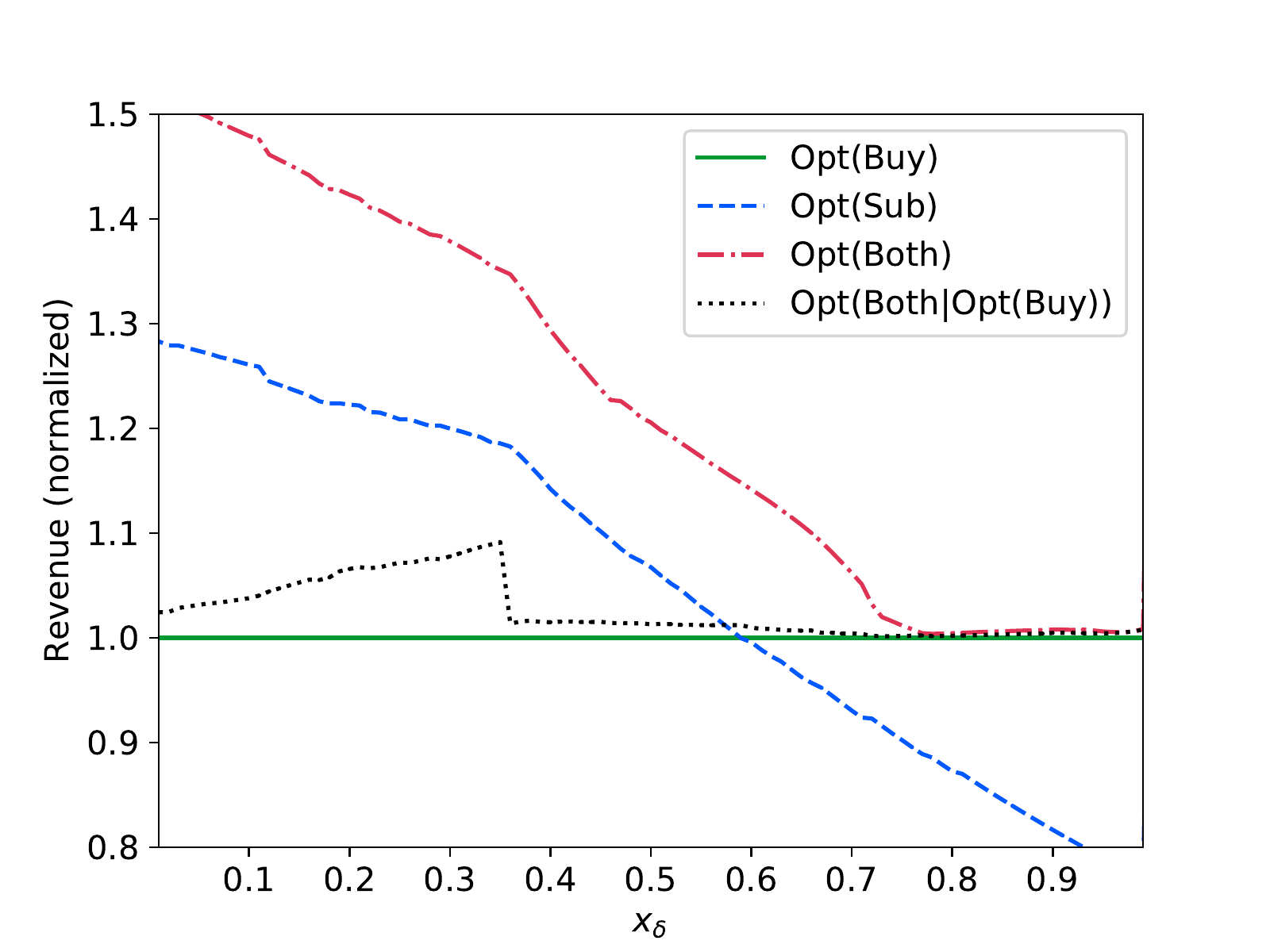}
                \caption{Revenue for different long-term\newline engagement distributions (i.e., varying  $x_\delta$)}
                \label{FIG:vart}
        \end{minipage}%
        \begin{minipage}{0.50\textwidth}
                \includegraphics[width=0.99 \textwidth]{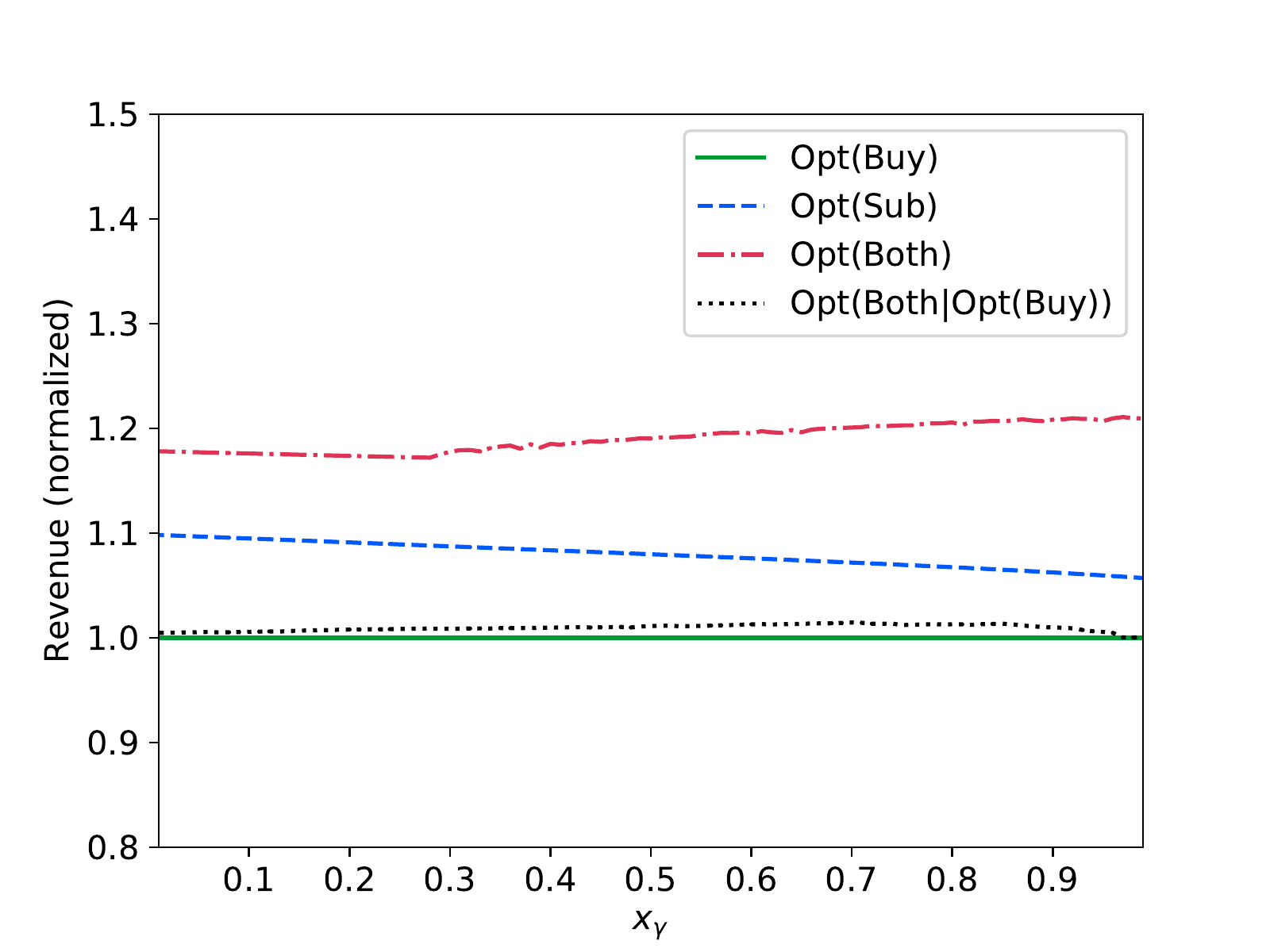}
                \caption{Revenue for different quality decay distributions (i.e., varying $x_\gamma$)}
                \label{FIG:varg}
        \end{minipage}%
       \end{figure}
In Figure \ref{FIG:vart}, we see that the  subscription option is best when the spread between short-term and long-term users is largest (i.e., $x_\delta$ is small), as subscription inherently differentiates users on how long they are interested in the product. For very high $x_\delta$, i.e., when every user in expectation has demand for many timesteps, the potential revenue gain of offering a subscription alongside a buy option becomes very small (on the magnitude of $0.5\%$). This is because at that point, too many users stop subscribing despite still having demand because their perceived value decayed below the subscription price. This makes subscription options inherently unattractive and any subscription price that could attract a large amount of additional users would cause too much market cannibalization.  Interestingly, for low $x_\delta$ ($<0.36$ when only offering a buy option and $<0.46$ when offering both options), it is optimal to give out the upgrade for free (but make the base product more expensive). This is whats causes the sudden drop in attainable revenue increase for adding a subscription option without changing the buy price, as users with relatively high quality decay that were priced out of buying before that point (and thus were willing to get a relatively expensive subscription for 1 or 2 time steps) afterwards switch to only buying the base game. 

In Figure \ref{FIG:varg}, we see that while increasing $x_\gamma$, and therefore decreasing the population variance of the quality decay factor $\gamma$, decreases the relative revenue potential of just offering a subscription, it can still increase the revenue potential of offering both options as market cannibalization decreases. 
       \begin{figure}[h!]
        \centering%
        \begin{minipage}{0.50\textwidth}
                \includegraphics[width=0.99 \textwidth]{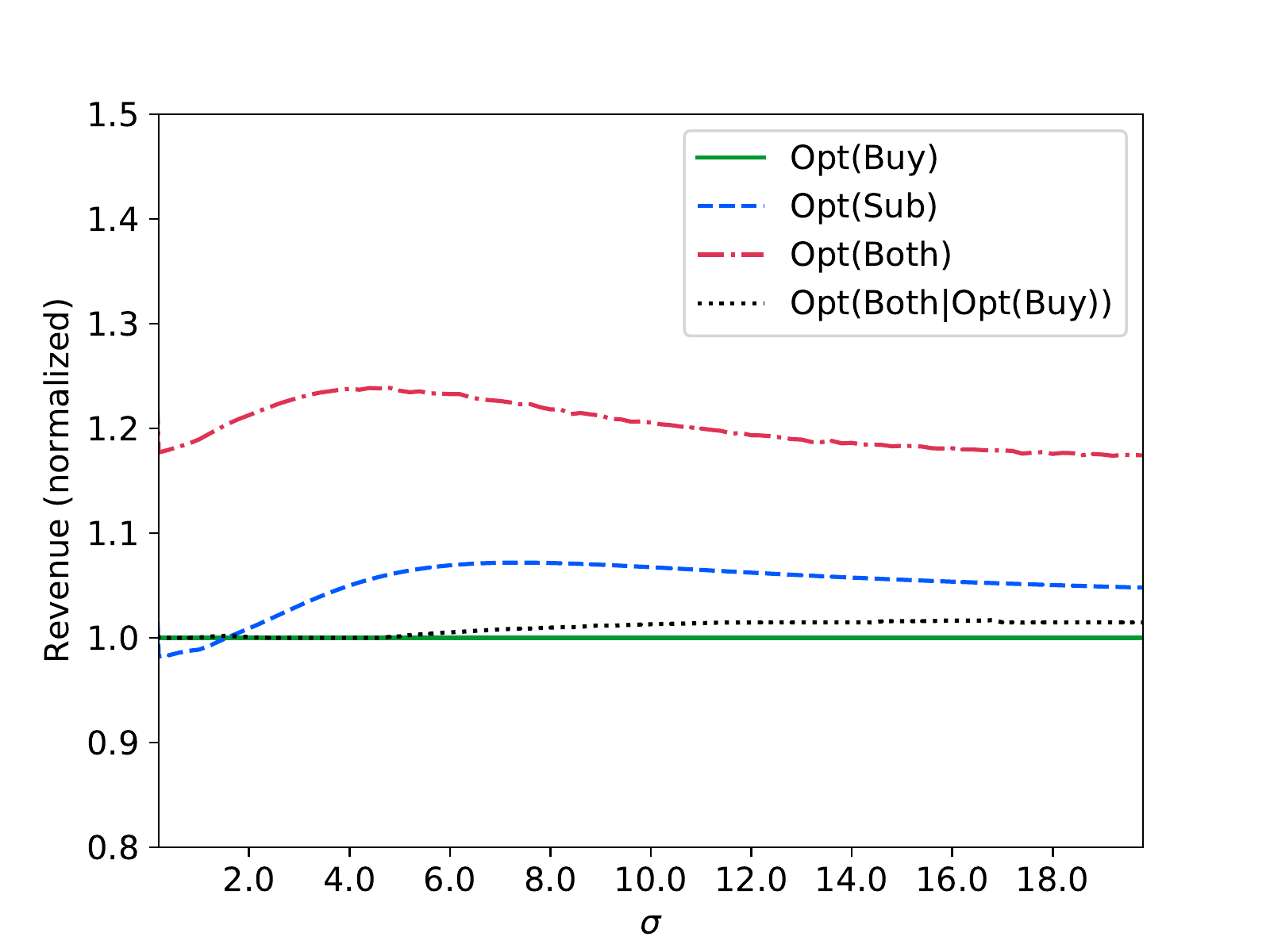}
                \caption{Revenue for different value\newline distributions (i.e., varying $\sigma$)}
                \label{FIG:varv}
        \end{minipage}%
        \begin{minipage}{0.50\textwidth}
                \includegraphics[width=0.99 \textwidth]{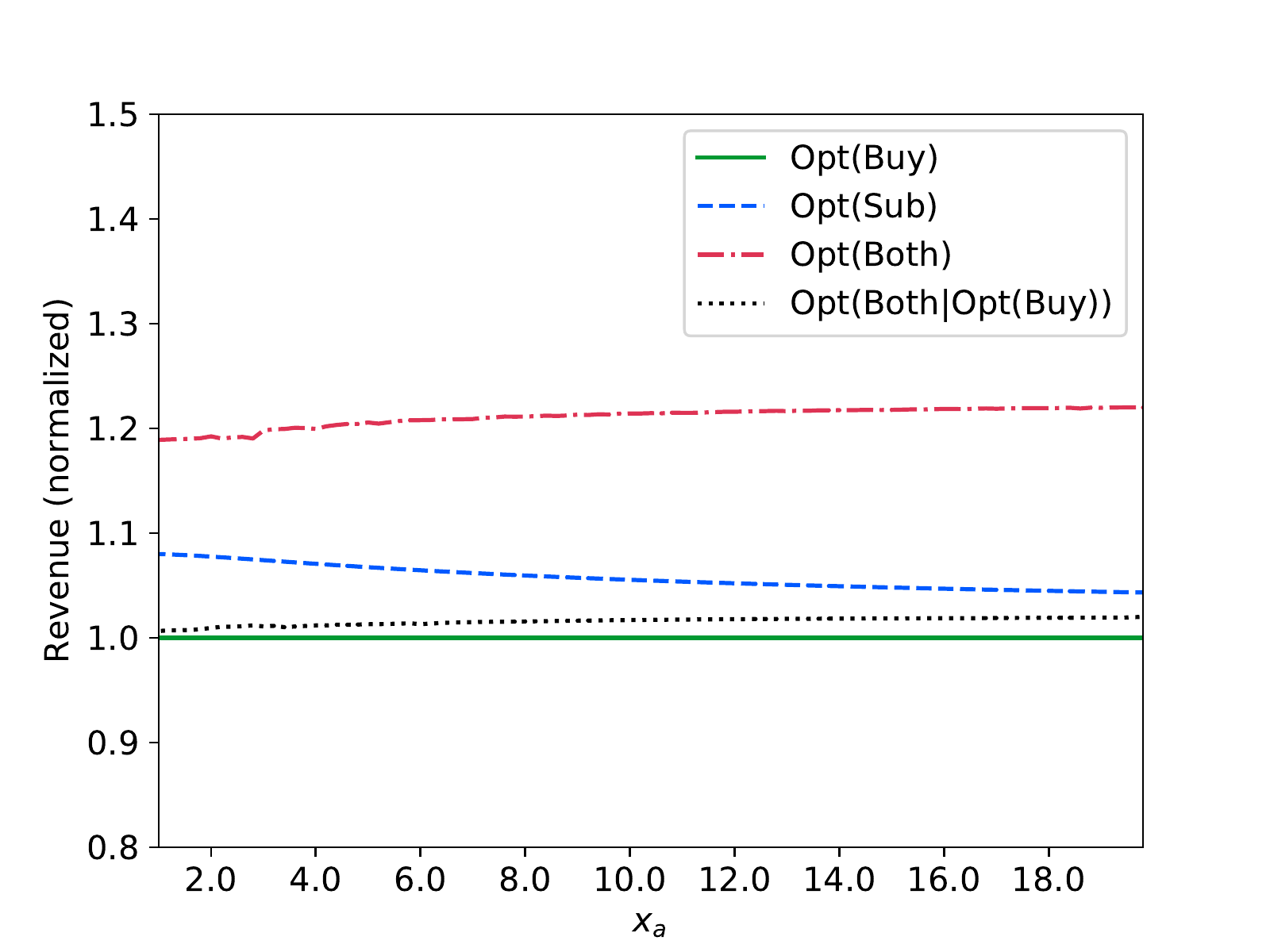}
                \caption{Revenue for different arrival\newline distributions (i.e., varying $x_a$)}
                \label{FIG:vara}
        \end{minipage}%
       \end{figure}
       
In Figure \ref{FIG:varv}, we see that the variance of the user values for the most part has a relatively low impact on the relative revenue potential of the different strategies. But since the optimal price when only offering a buy option for with the given parameters is low enough that users with average valuation buy even if the have a low long-term engagement factor $\delta$ and decay factor $\gamma$, the potential revenue advantage of offering a subscription option without changing buy prices goes to zero with low $\sigma$. 

    In Figure \ref{FIG:vara}, we see that how many users arrive in timestep $1$ also has a relatively low impact on the relative revenue potential of the different strategies. Similar to low $\sigma$, we again see that for very low $x_a$, the potential revenue improvement of offering a subscription option without changing the buy price goes to zero. Here this is caused by the fact that the buy price keeps decreasing because late arriving users become relatively more important for the publisher's revenue, making it harder and harder to offer a reasonably prices subscription option without loosing revenue to market cannibalization. 
       
\subsection{Correlating value and long term demand}
  \begin{figure}[h!]
        \centering%
        \begin{minipage}{0.50\textwidth}
                \includegraphics[width=0.99 \textwidth]{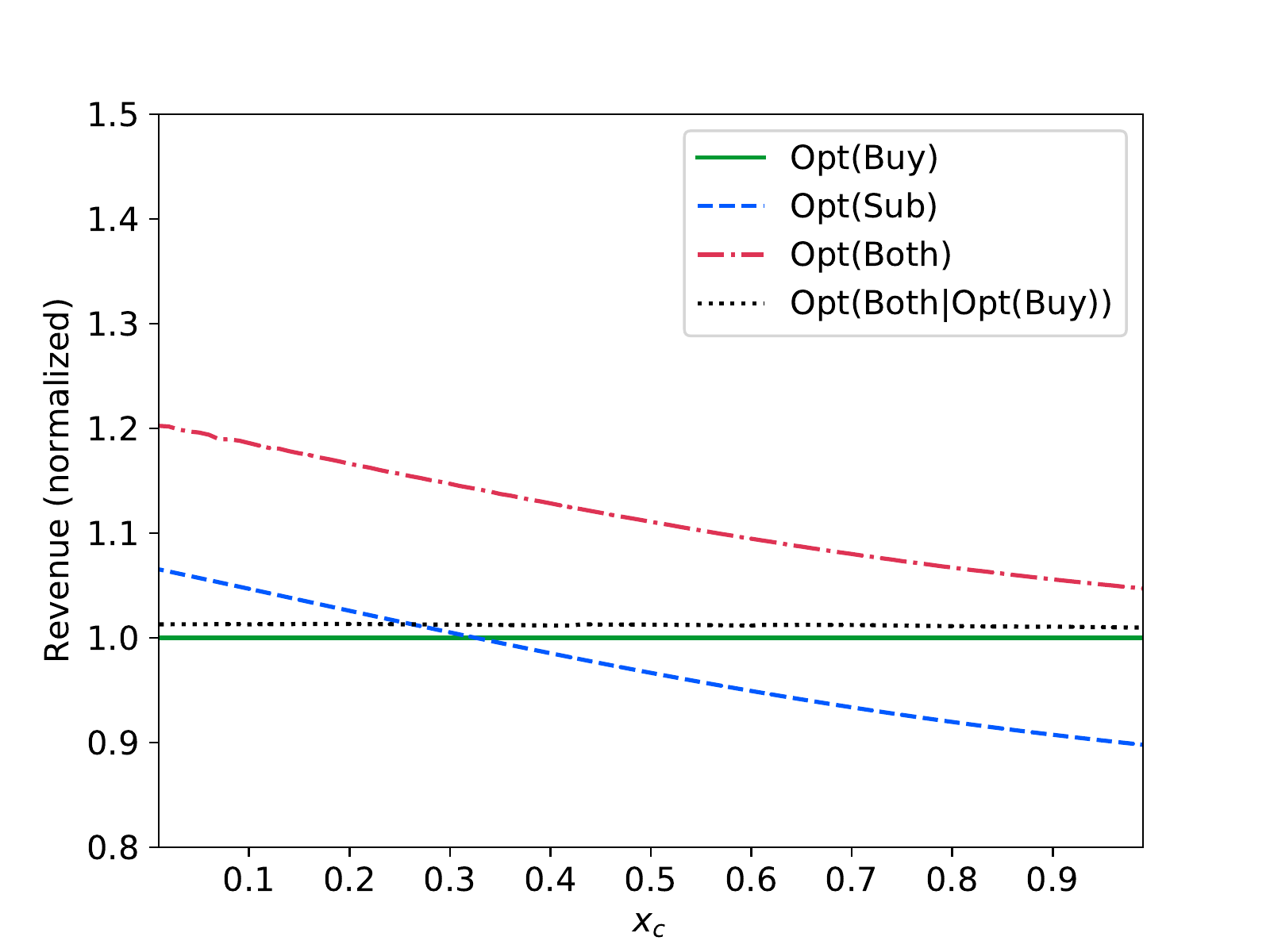}
                \caption{Revenue for different levels of\newline correlation (i.e., varying $x_c$)}
                \label{FIG:varc}
        \end{minipage}%
  \end{figure}
So far all type variables were assumed to be independent from each other. In practice it is likely that some (but not all) users with high long term engagement factor $\delta$ have lower values $v$, for example because they do not have much leisure time to use the game in each timestep and therefore need to own it for longer to spend the same time using it. In this example we want to study how introducing some correlation between $\delta$ and $v$ changes the publisher's revenue. To that end, we let the value distribution depend on $\delta$. For a dependence factor of $x_c$ set the value distribution for a users with long-term engagement factor $\delta$ to a normal distribution with mean $\mu= 25((1-x_c)+x_c(1-\delta))$ and standard deviation $\sigma = 10$, again truncated to $\left[0,50\right]$. 
As we can see in Figure \ref{FIG:varc}, while increasing the dependence between value and long-term engagement decreases the revenue potential of a subscription option, offering both options is still markedly better even when the mean of the normal distribution underlying a users value is fully dependent on the user's long-term engagement factor (i.e., $x_c = 1$). Interestingly, this dependence does not seem to have much, if any, effect on the revenue potential of introducing a subscription option without changing the buy prices.
\section{Conclusion}
We have analyzed the revenue maximization problem of a publisher wanting to either sell perpetual or subscription licenses for consumer software. 
In conclusion, combining subscription and perpetual license is typically revenue optimal when selling consumer software, realistically increasing revenue by $10\%-20\%$ over only offering perpetual license. Offering both types of licenses, it is often further possible to combine a revenue increase compared to only offering perpetual licenses with a Pareto improvement for the users, though the resulting revenue increase is then only on the magnitude of $1\%-2\%$. 


%
%
%
 \bibliographystyle{splncs04}
\bibliography{wine}
\newpage
\appendix
\section{Reward and payment with the potentially optimal strategies}\label{AP1}
In this section we derive expressions for  the expected reward and payment for any of the potentially optimal strategies that can be evaluated relatively cheaply. 
Towards this, we first derive the following technical lemma. 

\begin{lemma}\label{LEM:3}
        Assume a user of type $\tau$ with arrival time $n_a<m$ that either, beginning in timestep $n'= n_a$, subscribes in any timestep before $n''\leq m$ and takes no other action or buys the product in timestep $n'=n_a$ (in which case $n''= m$). The expected normalized reward such a user obtains \emph{before timestep $m$} is given by \begin{align}
        w^{<m}(n', n'', o_1) &= q_1 (\gamma^{n'}\frac{1- (l t \gamma)^{max(0,n''-n')}}{1- l t \gamma}+\gamma^{n''} o_1 \frac{1- (l t \gamma)^{m-n'') }}{1- l t \gamma}) 
        \end{align}
        The probability that such a user  has demand $d=1$ in timestep $m$ is given by 
        \begin{align}
        \kappa(n',n'',o)=
        \delta^{n''-n'}+ \delta(1-\delta^{n''-n'})
        \end{align}
        Similarly, assume a user of type $\tau$ with value $v=1$ and ownership vector $o$ that has demand in timestep $n' = max(n_a,m)$ and subscribes from $n'$ in any timestep before $n''\geq m$ and takes no other action. The expected value such a user obtains after timestep $m$ is given by        
        \begin{align}   
        w^{\geq m}(n', n'', o) =& (\gamma^{n'} q_1+ \gamma^{n'-m}q_2) \frac{1- (l t \gamma)^{max(0,n''-n')}}{1- l t \gamma}\\
        &+ 
        (\gamma^{n''} q_1 o_1 + \gamma^{n''-m}q_2 o_2)\frac{1}{1- l t \gamma}
        \end{align} 
        
        A user of type $\tau$ that has demand $d=1$ in timestep $n'$ and subscribes from $n'$ in any timestep before $n''\geq m$ where he still has demand and takes no other action makes an expected payment of 
        \begin{align}   
        \rho(n', n'') =& p_S \frac{1- (l t )^{n''-n'}}{1- l t}
        \end{align} 
\end{lemma}     

\begin{proof}

        Recall that the normalized reward of a user of type $\tau$ that follows some strategy $\alpha$ is given by
        \begin{align}w(\alpha,\tau)
        =& \sum_{n=n_a}^\infty \sum_{\sigma'} P(\sigma_n= \sigma'|\alpha, \tau) w_n(S,b,\tau, \sigma')
        .\end{align} 
        
        For a user with ownership vector $o$ that subscribes from $n'= n_a$ to $n''-1$ and takes no other action (i.e., $o$ never changes), the expected normalized reward before timestep $m$ is consequently given by 
        \begin{align}w^{<m}(n', n'', o_1)
        =& \sum_{n=n_a}^{m-1} \sum_{\sigma'} P(\sigma_n= \sigma'|(n', n''), \tau) w_n(S,b,\tau, \sigma')\\
        =& \sum_{n=n_a}^{m-1}  P(d_n= 1|\alpha, \tau) w_n(S,b,\tau, (1, o))
        .\end{align} 
        For any timestep $n$ in which the user subscribes, the probability to loose demand is $1-\delta$ and it follows 
        \begin{align}\sum_{n=n_a}^{n''-1}  P(d_n= 1|\alpha, \tau) w_n(S,b,\tau, (1, o))
        =&\sum_{n=n_a}^{n''-1}  \delta^{n-n_a} w_n(S,b,\tau, (1, o))\\
        &= \sum_{n=n_a}^{n''-1}  \delta^{n-n_a} \gamma^{n} q_1\\
        &= \sum_{n=n_a}^{n''}  (\gamma \delta)^{n-n_a} \gamma^{n_a} q_1\\
        &= \gamma^{n_a} q_1 \frac{1-(\gamma \delta)^{n''-n_a}}{1-\gamma \delta}
        .\end{align} 
        Here, the last equality follows as a partial geometric series. Analogously, taking into account that for $o_1=0$ no more value is obtained,  for the remaining timesteps it holds  
        \begin{align}\sum_{n=n''}^{m-1}  P(d_n= 1|\alpha, \tau) w_n(S,b,\tau, (1, o))
        &=  \gamma^{n''} o_1 q_1 \frac{1-(\gamma \delta)^{m-n''}}{1-\gamma \delta}
        .\end{align} 
        and the statement for $w^{<m}(n', n'', o_1)$ follows. 
        
        The probability that such a user still has demand in timestep $n''$ is given by $\delta^{n''-n'}$. If he does, then he also still has demand in timestep $m$. If he on the other hand does not have demand anymore in timestep $n''$, then he still has a probability of $\delta$ to regain demand with the upgrade release and price change in timestep $m$. 
        The statement for $w^{\geq m}(n', n'', o)$ follows analogously to the statement for $w^{< m}(n', n'', o_1)$.
        
\end{proof}

For all potentially optimal $\alpha$, we can now easily derive a user's reward and payment. 

\begin{proposition}\label{prop:rewpay}
        For strategy $\alpha = (\alpha_1,\alpha_2)$ with $\alpha_1 \in \left\lbrace b,s\right\rbrace$, $\alpha_2 \in \left\lbrace b,s, b_b\right\rbrace$, the normalized expected reward for playing $\alpha$ is
        \begin{align}   
        w((b,b),\tau)=& \begin{cases}
        \psi_1(n_a, n_a, 1) + \kappa(n_a,n_a, 1) \psi_2(m, m, [1,1])  & \;\;\; \text{if } n_a<m \\
        \psi_2(n_a, n_a, [1,1]) & \;\;\; \text{if } n_a\geq m
        \end{cases}     \\                                                                      
        w((b,s),\tau)=  &\begin{cases} 
        \psi_1(n_a, n_a, 1) + \kappa(n_a,n_a,1) \psi_2(m, n_2^{1,\tau}, [1,0]) & \;\;\; \text{if } n_a<m \\
        \psi_2(n_a, n_2^{1,\tau}, [1,0])&\;\;\; \text{if } n_a\geq m
        \end{cases}     \\
        w((s,s),\tau)   =&\begin{cases}
        \psi_1(n_a, n_1, 0)  + \kappa(n_a,n_1^\tau,0) \psi_2(m, n_2^{0,\tau}, [0,0])& \;\;\; \text{if } n_a<m \\
        \psi_2(n_a, n_2^{0,\tau}, [0,0])&\;\;\; \text{if } n_a\geq m
        \end{cases}     \\                                                                                                                              
        w((s,b),\tau)   =& \begin{cases} 
        \psi_1(n_a, n_1^\tau, 0) + \kappa(n_a,n_1^\tau,0) \psi_2(m, m, [1,1])& \;\;\; \text{if } n_a<m \\
        \psi_2(n_a, n_a, [1,1])&\;\;\; \text{if } n_a\geq m
        \end{cases}     \\
        w((s,b_b))      =&\begin{cases} 
        \psi_1(n_a, n_1^\tau, 0) +   \kappa(n_a,n_1^\tau,0) \psi_2(m, n_3^\tau, [1,0])& \;\;\; \text{if } {n_a}<m \\
        \psi_2(m, n_3^\tau, [1,0])
        &\;\;\; \text{if } n_a\geq m
        \end{cases}                                             
        \end{align}

        The  expected payments are
        \begin{align}   
        \rho((b,b),\tau)=& \begin{cases}
        p_1^{<m}+ \kappa(n_a,n_a,1) p_2& \;\;\; \text{if } n_a<m \\
        p_1^{\geq m}+p_2&\;\;\; \text{if } n_a\geq m
        \end{cases}             \\                                                                      
        \rho((b,s),\tau)=       &\begin{cases} p_1^{<m}+ \kappa(n_a,n_a,1)\rho_S(n_m, n_2^{1,\tau})& \;\;\; \text{if } n_a<m \\
        p_1^{\geq m}+\rho_S(n_a, n_2^1)&\;\;\; \text{if } n_a\geq m
        \end{cases}     \\
        \rho((s,s),\tau)        =& \begin{cases}  \rho_S(n_a, n_1^\tau) +\kappa(n_a,n_1^\tau,0) \rho_S(m, n_2^{0,\tau})
        & \;\;\; \text{if } n_a<m \\
        \rho_S(n_a, n_2^0)&\;\;\; \text{if } n_a\geq m
        \end{cases}             \\                                                                                                                      
        \rho((s,b),\tau)        =& \begin{cases}   \rho_S(n_a, n_1^\tau)+\kappa(n_a,n_1^\tau,0) (p_1^{\geq m}+p_2)& \;\;\; \text{if } n_a<m \\
        p_1^{\geq m}+ p_2&\;\;\; \text{if } n_a\geq m
        \end{cases}             \\      
        \rho((b,b_b),\tau)      =& \begin{cases} \rho_S(n_a, n_1^\tau) +\kappa(n_a,n_1^\tau,0)  \left(\rho_S(m, n_3^\tau)+\delta^{n_3^\tau-m} p_1^{\geq m} \right)& \text{if } n_a<m \\
        \rho_S(n_a, n_3^\tau)+\delta^{n_3^\tau-n_a} p_1^{\geq m}& \text{if } n_a\geq m
        \end{cases}                                                             
        \end{align}
\end{proposition}

\begin{proof}
        
        The statement follows by combining the actions taken under each strategy as described in Lemmas \ref{LEM:1} and \ref{LEM:2} with Lemma \ref{LEM:3} by noting that the player always either subscribes, owns all the available products he plans to buy with his strategy or does not have demand.
\end{proof}
\newpage
\section{Owners, Prices and Activity Data Examples}\label{AP2}
\subsection{Stellaris, a  4X grand strategy strategy game} 
This is a so called 4X (Explore, Expand, Exploit, Exterminate) grand strategy game with regular small  updates (which we do not model) and optional paid upgrades. During the observed timeframe, 4 large paid upgrades released, in October of 2016, April and September of 2017 , as well as frebruary of 2018. Each coincides with an up-tick in active players, though only the April upgrade seems to have lead to a large rise in sales. This most likely happened because it brought the quality of the whole product above a level where users that had held up on buying finally bought the base product.  
 \begin{figure}[h!]
        \centering%
        \begin{minipage}{0.5\textwidth}
                \includegraphics[width=0.99 \textwidth]{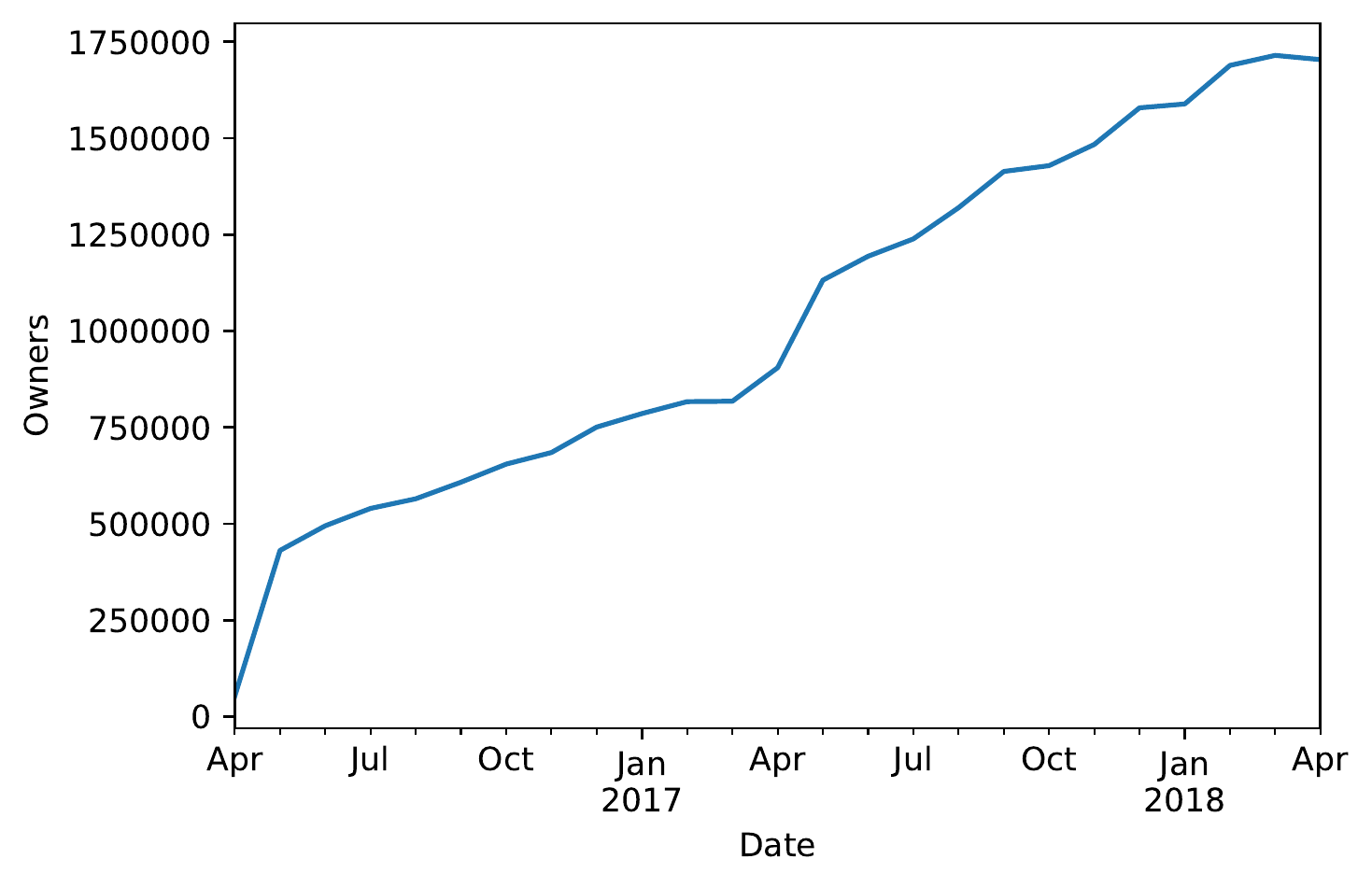}
                \caption{Stellaris:Owners over time}
                \label{FIG:stelown}
        \end{minipage}%
        \begin{minipage}{0.5\textwidth}
                \includegraphics[width=0.99 \textwidth]{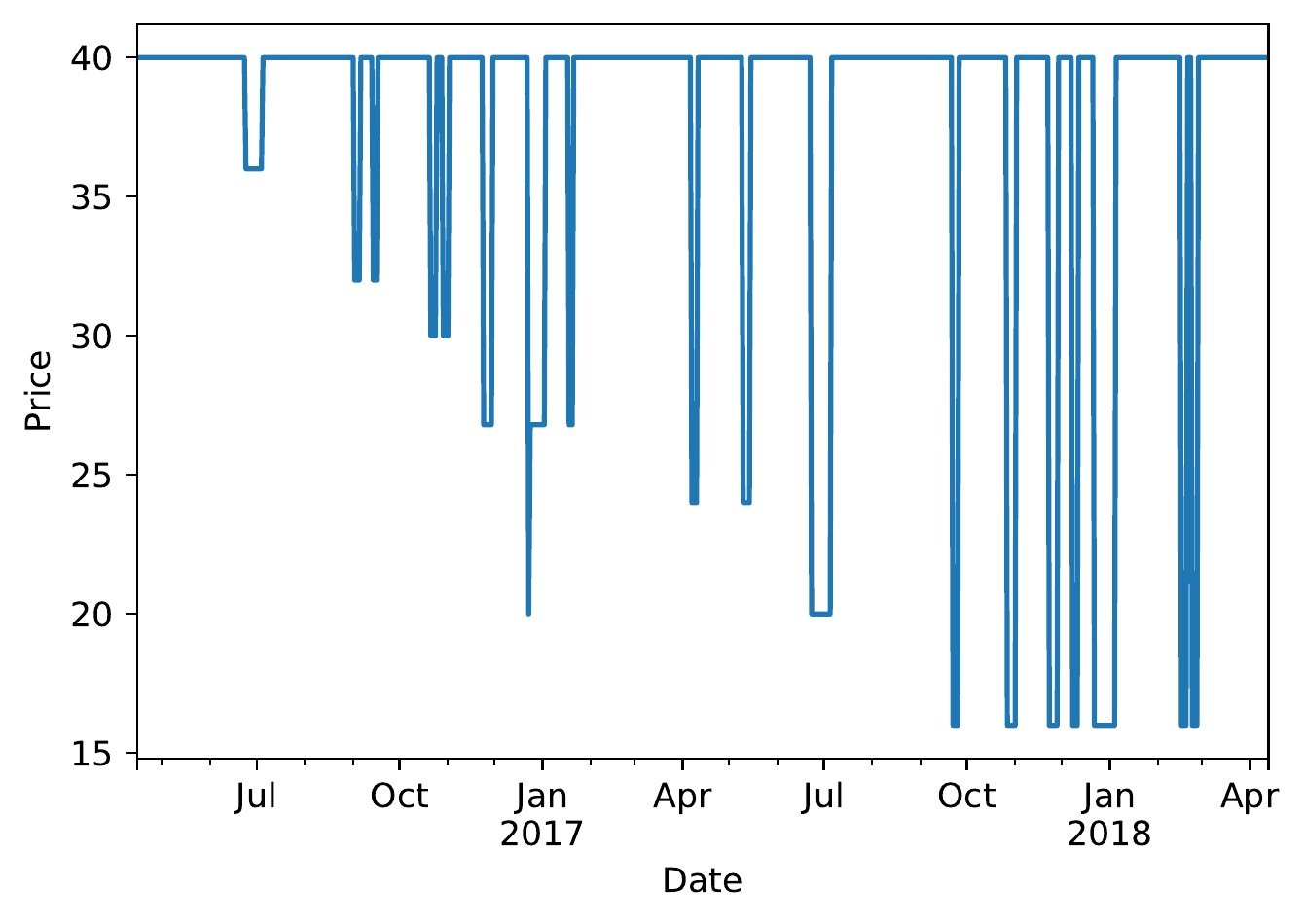}
                \caption{Stellaris:Price over time}
                \label{FIG:stelprice}
        \end{minipage}%
        
 \end{figure}
 \begin{figure}[h!]
        \centering%
        \begin{minipage}{0.5\textwidth}
                \includegraphics[width=0.99 \textwidth]{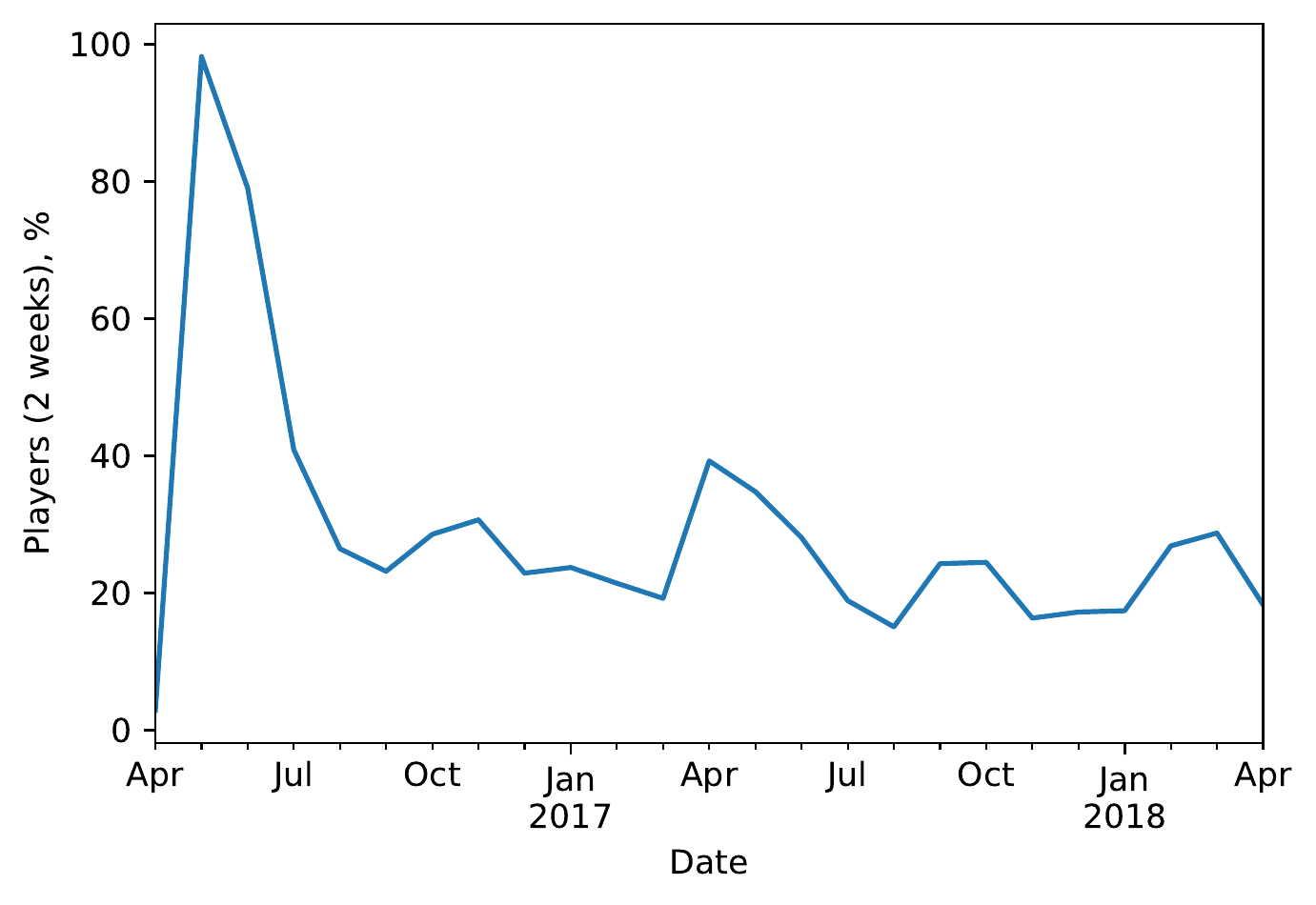}
                \caption{Stellaris: recently active owners in $\%$}
                \label{FIG:stelret}
        \end{minipage}%
 \end{figure}
 \newpage
\subsection{Dark Souls III, an action RPG}
This is a large, story driven game with a minor multi-player component. Famous for being very challenging. There were two paid upgrades released, one in October 2017 and one in March of 2017 that are both visible as up ticks in active players.

  \begin{figure}[h!]
        \centering%
        \begin{minipage}{0.5\textwidth}
                \includegraphics[width=0.99 \textwidth]{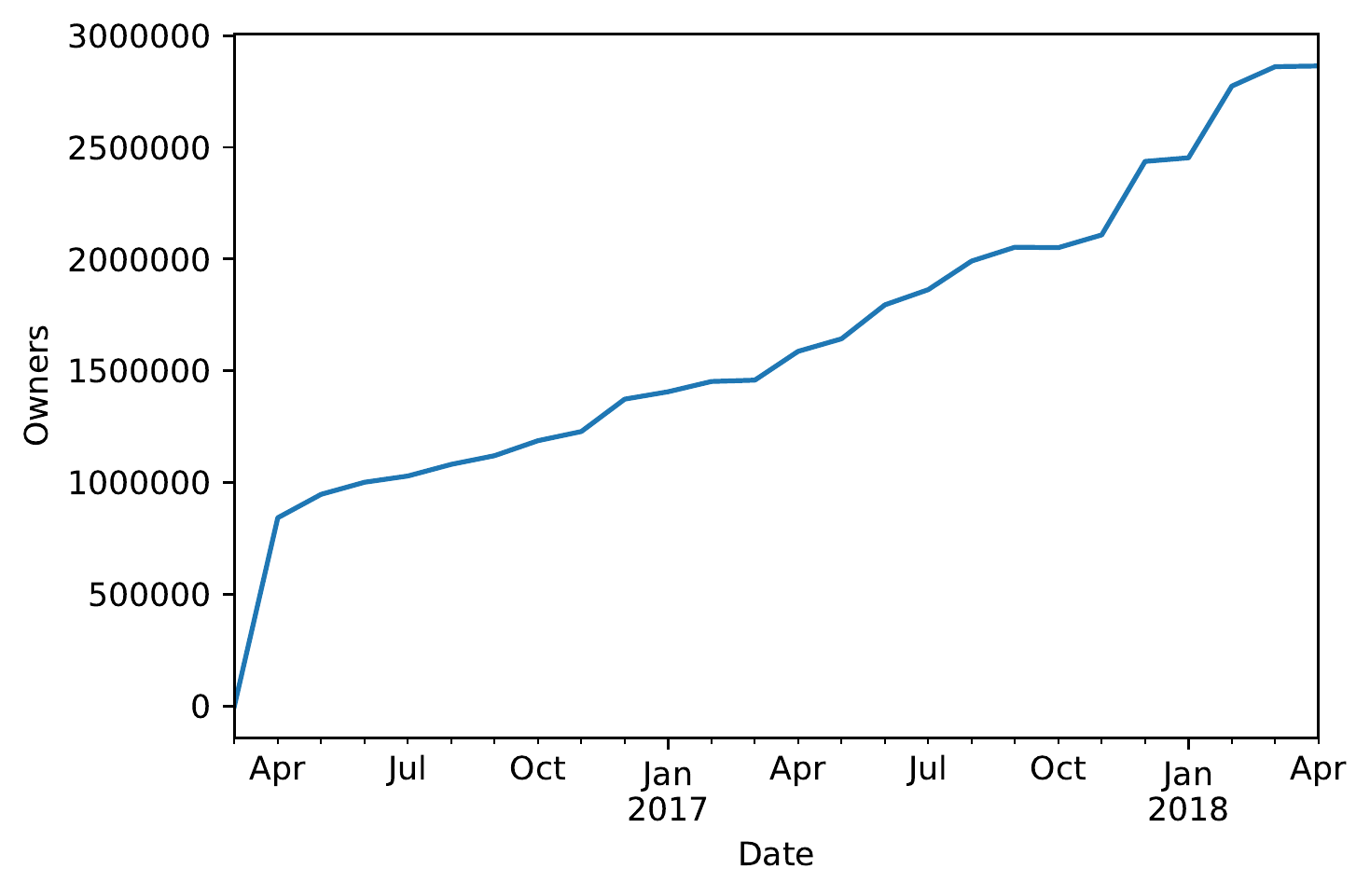}
                \caption{Dark Souls III:Owners over time}
                \label{FIG:DSown}
                \end{minipage}%
                \begin{minipage}{0.5\textwidth}
                        \includegraphics[width=0.99 \textwidth]{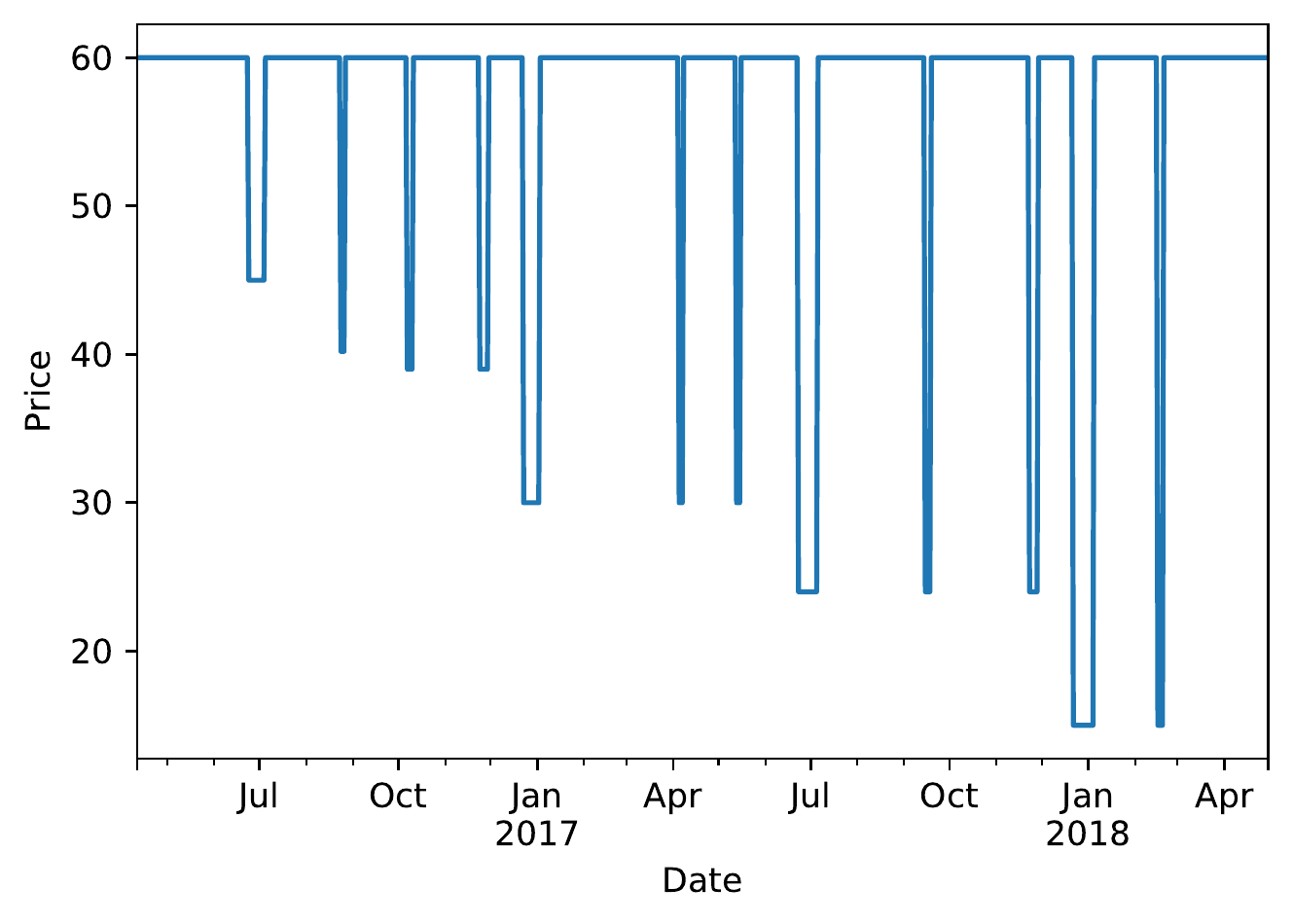}
                        \caption{Dark Souls III:Price over time}
                        \label{FIG:DSprice}
                        \end{minipage}%

        \end{figure}
  \begin{figure}[h!]
        \centering%
        \begin{minipage}{0.5\textwidth}
                \includegraphics[width=0.99 \textwidth]{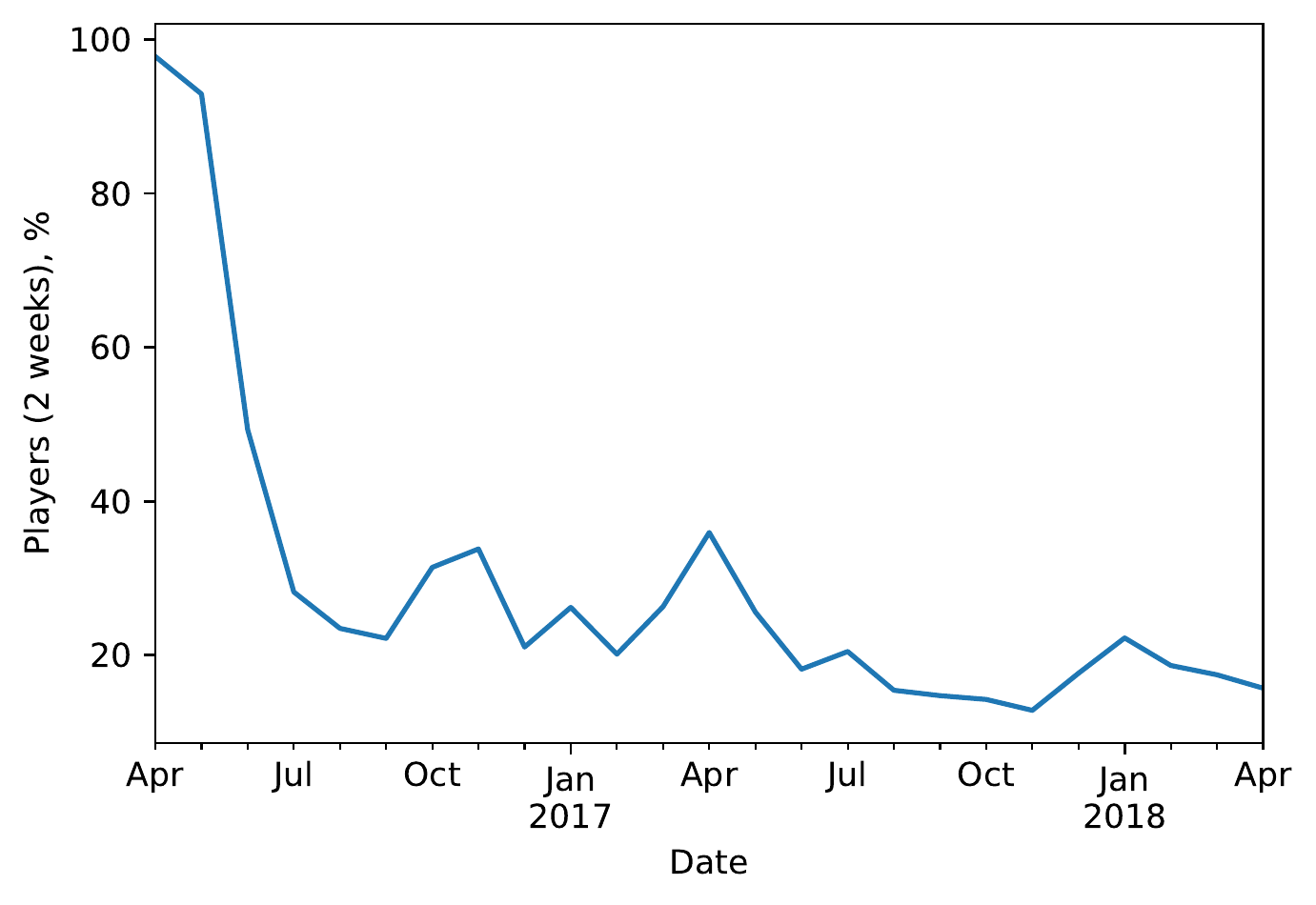}
                \caption{Dark Souls III: recently active owners in $\%$}
                \label{FIG:DSret}
        \end{minipage}%
\end{figure}
\newpage

\subsection{Slay the Spire, a rogue-like card deck building game}
This is a relatively small game that user typically play with low intensity, but for a long period of time. Additionally, in difference to the other two games, Slay the Spire had an so called 'early access' period during which an unfinished version was sold at a lower price while obtaining regular free updates. In our model, this in roughly comparable to \emph{increasing} (instead of decreasing) the base products price when the upgrade releases, but giving out the upgrade for free.  
  \begin{figure}[h!]
        \centering%
        \begin{minipage}{0.5\textwidth}
                \includegraphics[width=0.99 \textwidth]{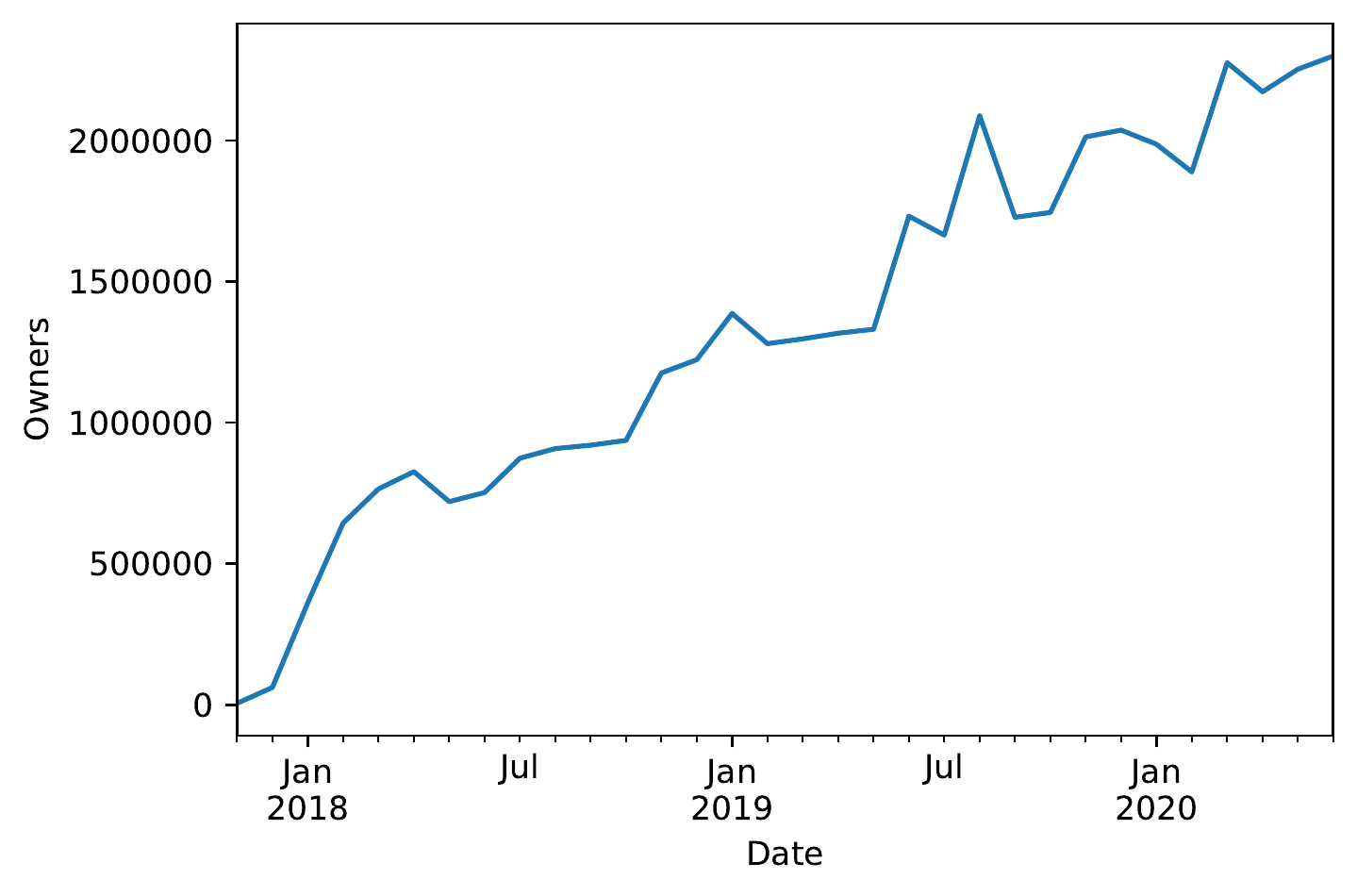}
                \caption{Slay the Spire: Owners over time}
                \label{FIG:SLown}
        \end{minipage}%
        \begin{minipage}{0.5\textwidth}
                \includegraphics[width=0.99 \textwidth]{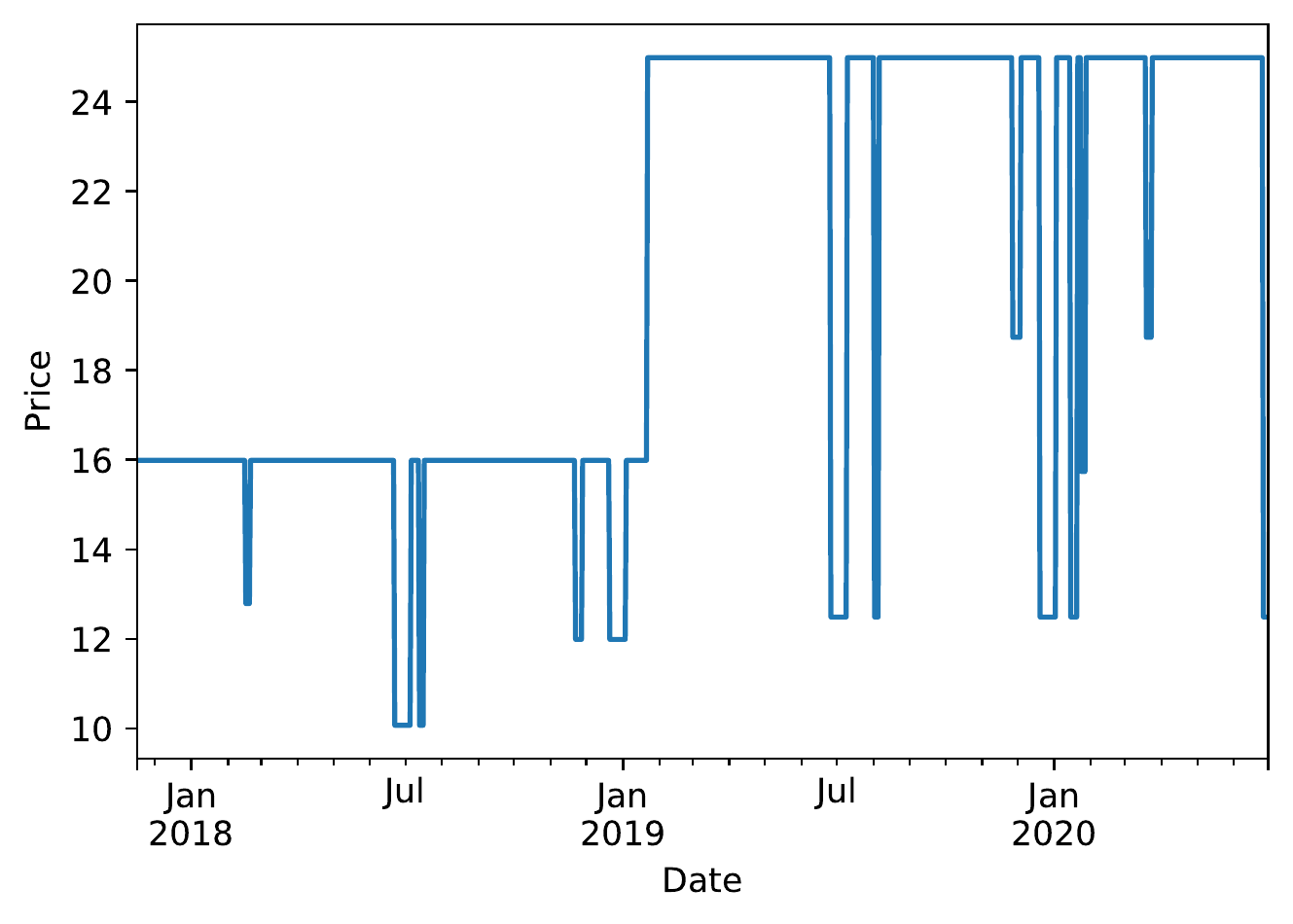}
                \caption{Slay the Spire: Price over time}
                \label{FIG:SLprice}
        \end{minipage}%
        
  \end{figure}
  \begin{figure}[h!]
        \centering%
        \begin{minipage}{0.5\textwidth}
                \includegraphics[width=0.99 \textwidth]{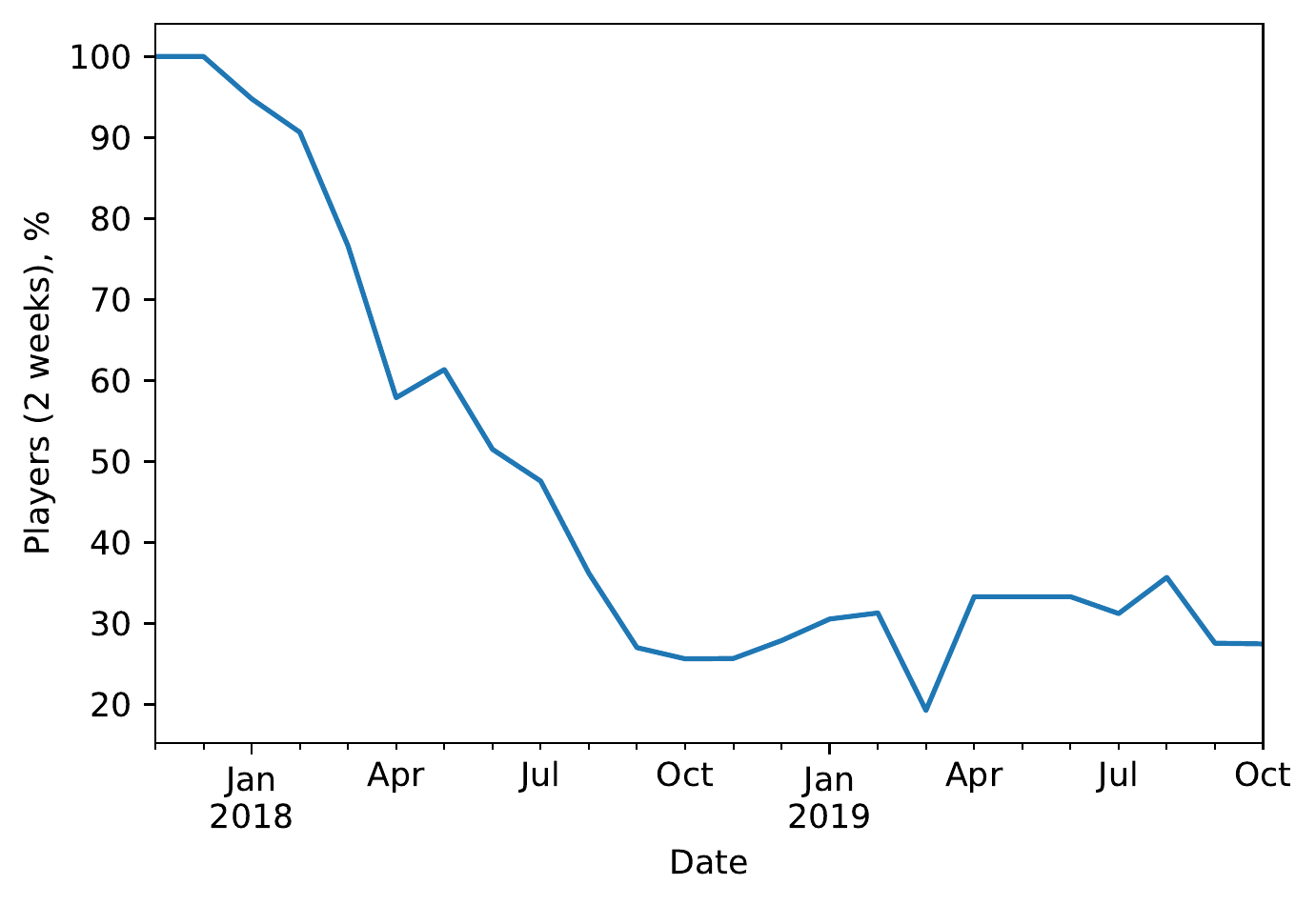}
                \caption{Slay the Spire: Recently active owners in $\%$}
                \label{FIG:SLren}
        \end{minipage}%
  \end{figure}
\end{document}